\documentclass[final,11pt,3p]{elsarticle}

\usepackage{graphicx}
\usepackage{amssymb}
\usepackage{amsthm}
\usepackage{amsmath}
\usepackage{algorithm}
\usepackage{algorithmic}
\usepackage{tikz}
\usepackage{lineno}
\usepackage{hyperref}
\usetikzlibrary{shapes}

\usepackage{subcaption}
\usepackage{caption}
\usepackage{float}
\usepackage{geometry}
\usepackage{mathrsfs} 
\usepackage{diagbox}
\biboptions{sort&compress}

\newtheorem{theorem}{Theorem}
\newtheorem{lemma}{Lemma}
\newtheorem{definition}{Definition}
\newtheorem{remark}{Remark}

\newtheorem{property}{Property}

\DeclareMathOperator{\Real}{Re}

 \newcommand{\Dat}{\mathcal{D}^{\alpha}_{t}}
 \newcommand{\ABC}{\mathcal{ABC}}
 
\usepackage{xcolor}
 
\begin{document}

\begin{tikzpicture}[remember picture,overlay]
	\node[anchor=north east,inner sep=20pt] at (current page.north east)
	{\includegraphics[scale=0.2]{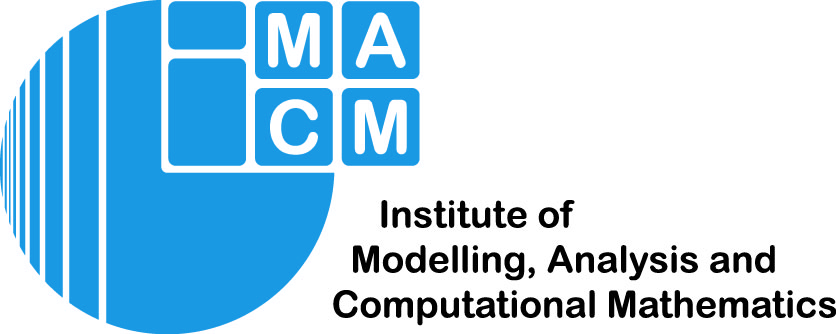}};
\end{tikzpicture}

\begin{frontmatter}


\title{Mathematical Modeling and Hyers-Ulam Stability for a Nonlinear Epidemiological Model with $\Phi_p$ Operator and Mittag-Leffler Kernel}

\author[AMNEA]{Achraf Zinihi}
\ead{a.zinihi@edu.umi.ac.ma}

\author[AMNEA]{Moulay Rchid Sidi Ammi}
\ead{rachidsidiammi@yahoo.fr}

\author[BUW]{Matthias Ehrhardt\corref{Corr}}
\cortext[Corr]{Corresponding author}
\ead{ehrhardt@uni-wuppertal.de}

\address[AMNEA]{Department of Mathematics, MAIS Laboratory, AMNEA Group, Faculty of Sciences and Technics,\\
Moulay Ismail University of Meknes, Errachidia 52000, Morocco}

\address[BUW]{University of Wuppertal, Chair of Applied and Computational Mathematics,\\
Gaußstrasse 20, 42119 Wuppertal, Germany}

\begin{abstract}
This paper investigates a novel nonlinear singular fractional SI model with the $\Phi_p$ operator and the Mittag-Leffler kernel. 
The initial investigation includes the existence, uniqueness, boundedness, and non-negativity of the solution. 
We then establish Hyers-Ulam stability for the proposed model in Banach space. 
Optimal control analysis is performed to minimize the spread of infection and maximize the population of susceptible individuals.
Finally, the theoretical results are supported by numerical simulations.
\end{abstract}


\journal{}
\date{February 19, 2024}

\end{frontmatter}

\section{Introduction}\label{S1}
Fractional-order models have attracted considerable interest from researchers in a wide variety of disciplines. Over the past two decades, these models have found applications in a wide variety of scientific and engineering fields, including modern physics, signal theory, control theory, hydrodynamics, viscoelastic theory, fluid dynamics, set theory, computer networks, biology, etc. 
Relevant literature on these topics can be found in the works \cite{Ghoshal2024, Hilfer2000, Huo2015, Maamar2024, Podlubny1998, Samko1993, Wang2023}.

Recently, several researchers have studied \textit{fractional differential equations} (FDEs) with singularities using various mathematical methods. For example, Bai and Qiu \cite{Bai2009} established the \textit{existence and uniqueness} (EU) of the solution to a nonlinear singular \textit{boundary value problem} (BVP) of FDEs using the Krasnoselskii and Leray-Schauder fixed point theorems. 
They also demonstrated applications to underscore their results. Agarwal, O’Regan and Staněk \cite{Agarwal2010} studied the EU for a singular fractional BVP using the Riemann-Liouville fractional derivative. 
Bai and Fang \cite{Bai2004} studied a singular nonlinear coupled system of FDEs, using Leray-Schauder and Krasnoselskii fixed point techniques for the $\mathrm{EU}$ of the solution. 
Vong \cite{Vong2013} studied FDEs with singularity and non-local boundary conditions using the Schauder fixed-point approach and upper-lower solution techniques. 
Pu et al.\ \cite{Pu2018} studied positive solutions of a multipoint BVP with singularity and applied their results to a specific example. 
Khan, Chen and Sun \cite{Khan2018} studied nonlinear FDEs with singularity and $p$-Laplacian to establish the EU of the solution and performed stability analysis.

Mathematical models have long been indispensable tools for understanding and predicting the dynamics of infectious diseases. 
By quantifying the complex interactions between pathogens and populations, these models enable researchers and policymakers to gain insight into the spread of disease and evaluate potential control strategies. 
Some of the pioneering work in epidemic modeling can be attributed to Kermack and McKendrick \cite{Kermack1927}, who introduced the SIR model in 1927. 
This model divided the population into three compartments: susceptible ($S$), infectious ($I$), and recovered ($R$). 
Using differential equations, the model captured the transitions between these compartments and laid the foundation for subsequent advances in epidemic modeling. 

Over time, mathematical models of epidemics have evolved and expanded to incorporate additional complexities. 
Researchers recognized the importance of accounting for factors such as age structure, spatial heterogeneity, and varying transmission rates. 
This led to the development of more sophisticated compartmental models, such as the SEIR model, which introduced an exposed compartment \cite{Anderson1991}. 
For a recent review on epidemiological models we refer the interested reader to \cite{NSFD_Review} and the references therein.

In addition, spatial epidemic models and network-based models emerged to capture the influence of geographic location and social connectedness on disease spread \cite{Keeling2007}. 
In recent years, advanced mathematical techniques have further enhanced the capabilities of epidemic models. 
Network theory has provided insights into the role of social connections in disease transmission, allowing the exploration of targeted intervention strategies. 
Nonlinear dynamics and chaos theory have shed light on complex epidemic behavior, including the emergence of periodic outbreaks and bifurcations.

The main goal of this paper is to minimize the number of infected individuals for the \textit{fractional SI model}, which describes the evolution of two compartments: susceptible individuals ($S$) and infected individuals ($I$), considering the \textit{Atangana-Baleanu fractional derivative in the Caputo sense} ($\ABC$ fractional derivative, for short) and the $\Phi_p$ operator.
The inclusion of these operators in our epidemic model offers several compelling motivations. 
First, the use of the $\ABC$ fractional derivative allows for the inclusion of memory effects in the model and long-range interactions in disease transmission.
Traditional derivative operators assume instantaneous changes, which may not accurately capture the dynamics of infectious diseases.
By introducing fractional calculus, we can account for the persistence of past infection rates, allowing for a more realistic representation of disease transmission.

Second, the $\Phi_p$ operator introduces nonlinearity and non-local interactions, reflecting the impact of infection on susceptible and infected populations within the SI model. 
This extension allows us to analyze the impact of localized outbreaks and potential hotspots within the epidemic dynamics.

This SI model has the potential to provide insight into the spread of infectious diseases under fractional order dynamics.
The $\ABC$ fractional derivatives and the $\Phi_p$ operator contribute to a more nuanced understanding of the interplay between susceptible and infected individuals, providing a valuable tool for designing effective disease control and mitigation strategies. 
The fractional nature of the system allows for the incorporation of memory effects, enhancing the model's realism and applicability to real-world scenarios. 
Further exploration of this coupled operator in this context provides an opportunity to delve into the intricate fractional calculus aspects of the model, contributing to a broader understanding of fractional-order epidemiological systems.

Nevertheless, we encountered challenges in establishing the positivity and boundedness of the solutions associated with the proposed model. 
These challenges arise from the nonlinearity introduced by the operator ${ }^{\ABC} \Dat [\Phi_p(\cdot)]$.
In addition, difficulties were encountered in determining the adjoint system, which is crucial for establishing the necessary optimality conditions.
At the same time, difficulties were encountered in numerical approximations aimed at determining the forward-backward $\ABC$ derivative.

Due to the distinct characteristics of these two operators, numerous experts and scholars have made significant contributions to the scientific literature in various fields.
For example, but not exhaustively, the work \cite{SidiAmmi2023} delves into the analysis of a reaction-diffusion SIR biological model formulated as a parabolic system of PDEs incorporating the $p$-Laplacian operator. 
This study emphasizes the critical role of vaccine distribution in the induction of immunity. 
The primary goal of this work is to develop an optimal control strategy that is carefully designed to limit both the spread of infection and the associated vaccination costs.

Jena, Chakraverty, and Baleanu \cite{Jena2021} used the homotopy perturbation Elzaki transformation method to derive solutions to an epidemic model of childhood diseases involving the $\ABC$ fractional derivatives. 
Their primary goal was to protect children from diseases preventable by vaccination.
They proposed the \textit{homotopy perturbation method} (HPM) transform to address the problem at hand because of the occasional challenges the Elzaki transform faces in handling nonlinear terms within the FDEs.
The study by \cite{Zhao2023} focused primarily on nonlinear coupling $(p_1, p_2)$ Laplacian systems with the nonsingular $\ABC$ fractional derivative. 
He established sufficient criteria for the existence and uniqueness of solutions based on the parameter values $p_1$ and $p_2$.
This interest and wide application of these two operators in various scientific disciplines underscores their broad and versatile influence on the research landscape.
Further insights and relevant literature on these topics can be explored in works such as \cite{Chai2012, Chen2012, Han2023, Hosseininia2022, Niimi2012}.

This paper is organized as follows. 
In Section~\ref{S2}, several key and crucial definitions are given. 
In Section~\ref{S3} we present the fractional optimal control SI model with $\Phi_p$ operator and Mittag-Leffler kernel. 
In Section~\ref{S4} we prove the existence, uniqueness, non-negativity, and boundedness of the solution. 
Section~\ref{S5} deals with the Hyers-Ulam stability of the proposed problem. 
Furthermore, Section~\ref{S6} is devoted to the determination of the necessary optimality conditions. 
Before concluding the present study, interesting numerical approximations are explained in Section~\ref{S7}. 
Finally, we conclude our study in Section~\ref{S8}.

\section{Notations and preliminaries}\label{S2}
In this section, we will recall some definitions and properties that we will need in the next sections on the Mittag-Leffler function, which plays an important role in the solution of fractional order differential and integral equations. To do this, we will introduce
the $\ABC$ fractional derivative,
the $\mathcal{AB}$ fractional integral, and the $\Phi_p$ operator. 
For the rest of this paper, we will choose
$\alpha\in (0,1)$ and $2 \le p<\infty$.

First, the two-parameter and one-parameter Mittag-Leffler functions are defined as follows.
\begin{definition}[Mittag-Leffler function, \cite{Atangana2016}]\label{D1}
Let $\omega \in \mathbb{C}$ such that $\Real(\omega)>0$, the Mittag-Leffler function $E_{\alpha, \beta}$ is defined by
\begin{equation}\label{MGMittag-Lefler}
     E_{\alpha,\beta}(\omega) = \sum_{k=0}^{\infty} \frac{\omega^{k}}{\Gamma(\alpha k+\beta)}, \quad\beta>0,
\end{equation}
where
\begin{equation*}
    \Gamma(\omega)=\int_0^{+\infty} e^{-t} t^{\omega-1}\,dt.
\end{equation*}
If $\beta =1$, then
\begin{equation*}
   E_{\alpha}(\omega)=E_{\alpha, 1}(\omega)=\sum_{k=0}^{\infty} \frac{\omega^{k}}{\Gamma(\alpha k +1)}.
\end{equation*}
\end{definition}
We note that \eqref{MGMittag-Lefler} was introduced by Magnus G\"osta Mittag-Leffler \cite{mittag1903generalisation} in 1903 with $\beta= 1$.
Next, we provide basic facts about the Gamma function and a link to the exponential function.
\begin{remark}\label{R1}\ \\[-.5cm]
\begin{itemize}
\item[i.] $\Gamma(\omega+1) = \omega\,\Gamma(\omega)$.

\item[ii.] By definition we have $\Gamma(1)=1$, and for any $n\in \mathbb{N}$ we find $\Gamma(n+1)=n!$.

\item[iii.] $E_{\alpha,\beta}$ is a generalization of the exponential function, and we have $E_{1}(\omega) = E_{1, 1}(\omega) = e^\omega$.
\end{itemize}
\end{remark}

We now define the $\ABC$ fractional derivative and the $\mathcal{AB}$ fractional integral. 
Let $T>0$, $f\in H^1(0,T)$ and $t\in (0,T)$. 
We set $B(\alpha) = (1-\alpha) + \frac{\alpha}{\Gamma(\alpha)}$ and $\gamma = \frac{\alpha}{1-\alpha}$.

\begin{definition}[$\ABC$ fractional derivative, \cite{Atangana2016}]\label{D2}\ \\[-.4cm]
\begin{itemize}
\item[a.] The $\ABC$ derivative of order $\alpha$ of $f$ with base point $0$ is defined at point~$t$ by
\begin{equation}\label{E2.1}
{ }^{\ABC} \Dat f(t)=\frac{B(\alpha)}{1-\alpha} \int_0^t f^{\prime}(y) E_\alpha [-\gamma (t-y)^\alpha]\,dy.
\end{equation}

\item[b.] The backward $\ABC$ derivative with base point $T$, is given by
\begin{equation}\label{E2.2}
{ }^{\ABC}_T \Dat f(t) 
= -\frac{B(\alpha)}{1-\alpha} \int_t^T f^{\prime}(y) E_\alpha
[-\gamma (y-t)^\alpha] \,dy.
\end{equation}
\end{itemize}
\end{definition}

Note that if we let $\alpha\to1$ in \eqref{E2.1}, then we get the usual derivative $\partial_t$.

\begin{definition}[$\mathcal{AB}$ fractional integral, \cite{Atangana2016}]\label{D3}
The $\mathcal{AB}$ fractional integral operator with base point $0$, is written as
\begin{equation}\label{E2.3}
{ }^{\mathcal{AB}} I^\alpha f(t) 
= \frac{1-\alpha}{B(\alpha)} f(t)+\frac{\alpha}{B(\alpha)\Gamma(\alpha)}\int_0^t f(s) (t-s)^{\alpha-1} \,ds.
\end{equation}
\end{definition}

The following Lemma~\ref{L1} describes the basic relation of
the $\ABC$ fractional derivative and the $\mathcal{AB}$ fractional integral. It will help us in the later Sections~\ref{S4} and \ref{S6} to prove the positivity and boundedness of the solution and to derive the adjoint system associated with the state model.
\begin{lemma}[\protect{\cite[Proposition 3.4.]{Abdeljawad2017}}]\label{L1}
Under the previous assumptions, we have 
\begin{equation}\label{E2.4}
{ }^{\mathcal{AB}} I^\alpha \bigl( { }^{\ABC} \Dat f(t)\bigr) = f(t) - f(0).
\end{equation}
\end{lemma}

The last part of this section deals with the $\Phi_p$ operator.
Let $\Phi_p$ be the function defined on $\mathbb{R}\to\mathbb{R}$ by
\begin{equation*}
    \Phi_p(w) = |w|^{p-2}w, \quad\forall w\in\mathbb{R}.
\end{equation*}
In addition, we have, 
\begin{equation*}
    \Phi_p^{-1}(\cdot) = \Phi_{p^*}(\cdot),\quad\text{where}\quad\frac{1}{p} + \frac{1}{p^*} = 1.
\end{equation*}

The two lemmas below will help us to prove that the proposed problem has a unique solution and to prove the associated Hyers-Ulam stability.
\begin{lemma}[\protect{\cite[Lemma~1.3]{Khan2018}}]\label{L2} \ \\[-.4cm]
\begin{itemize}
\item[i.] Let $x, y\in\mathbb{R}$ such that $|x|,|y|\le k$, then
\begin{equation}\label{E2.5}
    |\Phi_p(x)-\Phi_p(y)| \le (p-1) k^{p-2} |x-y|.
\end{equation}

\item[ii.] If $1<p\le 2$ and $x,y>0$ such that $|x|,|y|\ge k>0$, then
\begin{equation}\label{E2.6}
   |\Phi_p(x)-\Phi_p(y)| \le (p-1) k^{p-2} |x-y|.
\end{equation}
\end{itemize}
\end{lemma}

\begin{lemma}\label{L3}
Let $\varphi$ be a continuous bounded function on $[0,T]$ into $\mathbb{R}$. Then
\begin{equation}\label{E2.7}
    { }^{\ABC} \Dat \bigl[\Phi_p(\varphi(t))\bigr] 
    \cdot \varphi(t) \ge \frac{p-1}{p}\, { }^{\ABC} \Dat \bigl(\Phi_p(\varphi(t)) \varphi(t)\bigr).
\end{equation}
\end{lemma}
\begin{proof}
The inequality \eqref{E2.7} can be rewritten as
\begin{equation*}
     \mathcal{X} := { }^{\ABC} \Dat \bigl[\Phi_p(\varphi(t))\bigr] \cdot \varphi(t) - \frac{p-1}{p} { }^{\ABC} \Dat \bigl(\Phi_p(\varphi(t))\varphi(t)\bigr) \ge0.
\end{equation*}
Since 
\begin{equation}\label{E2.8}
   \bigl[\Phi_p(\varphi(t))\bigr]^{\prime} 
    = \bigl[|\varphi(t)|^{p-2} \varphi(t)\bigr]^{\prime}
    = \Bigl[e^{(p-2) \ln|\varphi(t)|} \varphi(t)\Bigr]^{\prime} 
    = (p-1) |\varphi(t)|^{p-2} \varphi^{\prime}(t),
\end{equation}
we have
\begin{equation*}
\bigl[\Phi_p(\varphi(t))\varphi(t)\bigr]^{\prime} 
= \bigl[\Phi_p(\varphi(t))\bigr]^{\prime}\varphi(t) + \Phi_p(\varphi(t))\varphi^{\prime}(t) 
= \frac{p}{p-1} \bigl[\Phi_p(\varphi(t))\bigr]^{\prime} \varphi(t).
\end{equation*}
Consequently,
\begin{equation*}
\begin{split}
\mathcal{X} &= { }^{\ABC} \Dat \bigl[\Phi_p(\varphi(t))\bigr] \cdot \varphi(t) - \frac{p-1}{p} { }^{\ABC} \Dat \bigl(\Phi_p(\varphi(t))\varphi(t)\bigr)\\
&= \frac{B(\alpha)}{1-\alpha} \biggl( \varphi(t) 
    \int_0^t \bigl[\Phi_p(\varphi(z))\bigr]^{\prime} 
   E_{\alpha}\bigl[-\gamma (t-z)^{\alpha}\bigr] \,dz - \frac{p-1}{p} 
  \int_0^t \bigl[\Phi_p(\varphi(z))\varphi(z)\bigr]^{\prime} 
   E_{\alpha}\bigl[-\gamma (t-z)^{\alpha}\bigr] \,dz\biggr)\\ 
&= \frac{B(\alpha)}{1-\alpha} \biggl(\varphi(t) 
    \int_0^t \bigl[\Phi_p(\varphi(z))\bigr]^{\prime} 
    E_{\alpha}\bigl[-\gamma (t-z)^{\alpha}\bigr] \,dz 
- \int_0^t \bigl[\Phi_p(\varphi(z))\bigr]^{\prime} 
   \varphi(z) E_{\alpha}\bigl[-\gamma (t-z)^{\alpha}\bigr]\,dz\biggr).
\end{split}
\end{equation*}
Thus,
\begin{equation*}
\begin{split}
\mathcal{X} &= \frac{B(\alpha)}{1-\alpha} 
\int_0^t \bigl(\varphi(t) - \varphi(z)\bigr)
\bigl[\Phi_p(\varphi(z))\bigr]^{\prime} 
E_{\alpha}\bigl[-\gamma (t-z)^{\alpha}\bigr] \,dz\\ 
&= \frac{B(\alpha)}{1-\alpha} \int_0^t
\Bigl(\int_z^t \varphi^{\prime}(s)\,ds\Bigr) \bigl[\Phi_p(\varphi(z))\bigr]^{\prime} E_{\alpha}\bigl[-\gamma (t-z)^{\alpha}\bigr]\,dz\\  
&= \frac{B(\alpha)}{1-\alpha} 
\int_0^t \varphi^{\prime}(s) 
\Bigl( \int_0^s \bigl[\Phi_p(\varphi(z))\bigr]^{\prime} 
     E_{\alpha}\bigl[-\gamma (t-z)^{\alpha}\bigr]\,dz\Bigr)\,ds\\
&\stackrel{\eqref{E2.8}}{=} \frac{B(\alpha)}{2(1-\alpha)(p-1)} 
\int_0^t |\varphi(s)|^{2-p} \bigl(E_{\alpha}[-\gamma (t-s)^{\alpha}]\bigr)^{-1} 
\frac{\partial}{\partial s}\biggl(\Bigl( 
\int_0^s \bigl[\Phi_p(\varphi(z))\bigr]^{\prime} 
E_{\alpha}\bigl[-\gamma (t-z)^{\alpha}\bigr]\,dz\Bigr)^2\biggr)\,ds.
\end{split}
\end{equation*}
Since $p\ge2$ and $\varphi$ is bounded, then $|\varphi(s)|^{p-2} \le C$ and that means $|\varphi(s)|^{2-p} \ge \sigma:=\frac{1}{C}$. Therefore,
\begin{equation*}
   \mathcal{X} \ge \frac{\sigma B(\alpha)}{2(1-\alpha)(p-1)} 
    \int_0^t \bigl(E_\alpha [-\gamma (t-s)^{\alpha}]\bigr)^{-1} 
    \frac{\partial}{\partial s}
    \biggl(\Bigl( \int_0^s \bigl[\Phi_p(\varphi(z))\bigr]^{\prime} 
    E_{\alpha}\bigl[-\gamma (t-z)^{\alpha}\bigr] \,dz\Bigr)^2\biggr) \,ds.
\end{equation*}
Knowing that 
\begin{equation*}
   \frac{\partial}{\partial s}
     \Bigl(\frac{1}{E_\alpha [-\gamma (t-s)^\alpha]}\Bigr) 
    = - \frac{\gamma (t-s)^{\alpha-1} E_{\alpha, \alpha} [-\gamma (t-s)^\alpha]}{E_\alpha[-\gamma (t-s)^\alpha]^2} \le0, \ \text{ and } \ E_\alpha(0) = 1.
\end{equation*}
Then, an integration by part gives
\begin{equation*}
\begin{split}
    \mathcal{X} \ge& \frac{\sigma B(\alpha)}{2(1-\alpha)(p-1)} 
    \biggl[\bigl(E_{\alpha}[-\gamma (t-s)^{\alpha}]\bigr)^{-1} 
    \Bigl( \int_0^s \bigl[\Phi_p(\varphi(z))\bigr]^{\prime} E_{\alpha}\bigl[-\gamma (t-z)^{\alpha}\bigr] \,dz\Bigr)^2 \biggr]_{s=0}^{s=t}\\
   &+ \frac{\sigma\gamma B(\alpha)}{2(1-\alpha)(p-1)}\int_0^t
    \frac{(t-s)^{\alpha-1} E_{\alpha, \alpha}[-\gamma (t-s)^{\alpha}]}{E_{\alpha}[-\gamma (t-s)^{\alpha}]^2} 
    \Bigl( \int_0^s \bigl[\Phi_p(\varphi(z))\bigr]^{\prime} E_{\alpha}\bigl[-\gamma (t-z)^{\alpha}\bigr] \,dz\Bigr)^2\,ds.
\end{split}
\end{equation*}
Thus,
\begin{equation*}
\begin{split}
\mathcal{X} &\ge \frac{\sigma B(\alpha)}{2(1-\alpha)(p-1)} 
\bigg[\Bigl( \int_0^s \bigl[\Phi_p(\varphi(z))\bigr]^{\prime} 
       E_\alpha \bigl[-\gamma (t-z)^{\alpha}\bigr] \,dz\Bigr)^2\\
&\quad + \int_0^t \frac{(t-s)^{\alpha-1} E_{\alpha, \alpha}[-\gamma (t-s)^{\alpha}]}{E_{\alpha}[-\gamma (t-s)^{\alpha}]^2} 
\Bigl( \int_0^s \bigl[\Phi_p(\varphi(z))\bigr]^{\prime} 
    E_{\alpha}\bigl[-\gamma (t-z)^{\alpha}\bigr] dz\Bigr)^2\,ds\bigg].
\end{split}
\end{equation*}
This concludes the proof.
\end{proof}

\section{Mathematical model with vital dynamics}\label{S3}
The SI model is a mathematical framework used to study the dynamics of an epidemic within a population. 
It divides the population ($N$) into different compartments based on their disease status. Let's look at the compartments and how the epidemic spreads from one to the other:
\begin{description}
    \item Susceptible $S$: This compartment represents individuals who are susceptible to the disease and can become infected if they come in contact with an infectious individual.

\item Infected $I$: This compartment includes individuals who are currently infected with the disease and can transmit it to susceptible individuals through various modes of transmission, such as respiratory droplets, physical contact, contaminated surfaces, etc.
\end{description} 

The following Table~\ref{Tab1} contains the transmission coefficients for the proposed SI model, and Figure~\ref{F1} summarizes how the epidemic spreads from one compartment to another.
\begin{table}[H]
\begin{minipage}{0.49\linewidth}
\centering
\setlength{\tabcolsep}{0.5cm}
\caption{Vital dynamics parameters of the SI model.}\label{Tab1}
\begin{tabular}{|c||c|}
\hline 
Symbol & Description\\
\hline \hline 
$\mu$ ($\ge0$) & Birth rate\\
\hline
$\beta$ ($\ge0$) & Effective contact rate\\
\hline
$\eta$ ($\ge0$) & Natural mortality rate\\ 
\hline
\end{tabular}
\end{minipage}
\hfill
\begin{minipage}{0.49\linewidth}
\centering
\begin{tikzpicture}[node distance=4cm]
\node (S) [diamond, draw, minimum size=0.8cm, fill=blue!30] {S}; 
\node (I) [diamond, draw, minimum size=0.8cm, fill=red!40, right of=S] {I};
\draw[->] (-1.4,0) -- ++(S) node[midway,above]{$\mu N$};
\draw[->] (S.south) -| (0,-1.2) node[near end,left]{$\eta S$};
\draw [->] (S) -- (I) node[midway,above]{$\beta SI$};
\draw[->] (I.south) -| (4,-1.2) node[near end,left]{$\eta I$};
\draw [->] (S.north) -| (0,1) -- (4,1) node[midway,above]{$u S$} -| (I.north) ;
\end{tikzpicture}

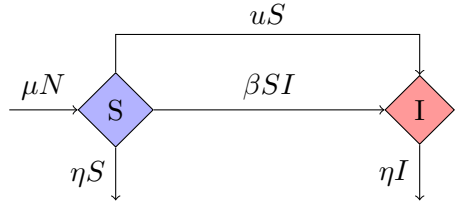
\captionof{figure}{Conversion scheme for the proposed SI model.}\label{F1}
\end{minipage}
\end{table}

Let $\mathcal{I}_T := [0,T]$. Given the assumptions explained above, our SI model is formulated as follows
\begin{equation}\label{E3.1}
\begin{cases}
 &{ }^{\ABC} \Dat \bigl[\Phi_p(S(t))\bigr] 
 = \mu N-\beta S I - uS - \eta S,\\
&{ }^{\ABC} \Dat \bigl[\Phi_p(I(t))\bigr] 
  = \beta S I + uS - \eta I,
\end{cases}  \quad t\in \mathcal{I}_T,
\end{equation}
with the total population $N(t)=S(t)+I(t)$ and supplied with the positive initial data 
\begin{equation}\label{E3.2}
   S(0)=S_0 \ge0, \text{ and } \ I(0) = I_0 \ge0.
\end{equation}
The term $\mu N - \beta S I - uS - \eta S$ in the first equation in \eqref{E3.1} represents the change in the susceptible population over time, influenced by infection (quadratic contact term) and external factors, where
\begin{equation}\label{E3.3}
     u \in \mathcal{U}_{ad} = \bigl\{u \text{ is Lebesgue integrable in } \mathcal{I}_T \ | \ 0\le u(t) < 1 \bigr\}.
\end{equation}
The variable $u$ in \eqref{E3.1} represents the optimal control (vaccination), while the term $uS$ denotes the vaccinated susceptible individuals. By this variable $u$, we want to minimize the spread of infection and maximize the population of susceptible individuals, cf.\ Section~\ref{S6}.

Let $\vartheta = (x,y)=(S,I)$ such that $\vartheta_0 = (x_0,y_0) = (S_0,I_0)$. We consider the vector function $\Psi$ defined by
\begin{equation*}
   \Psi(\vartheta(t)) = \bigl(\Psi_1(\vartheta(t)), \Psi_2(\vartheta(t))\bigr),
\end{equation*}
where
\begin{equation}\label{E3.4}
\begin{cases}
   \Psi_1(\vartheta(t)) = \mu(x(t) + y(t)) - \beta x(t)y(t) - u(t)x(t) - \eta x(t),\\
   \Psi_2(\vartheta(t)) = \beta x(t)y(t) + u(t)x(t) - \eta y(t).
\end{cases} \ t\in \mathcal{I}_T,
\end{equation}
with the initial condition $\vartheta_0 = (x_0,y_0)$. 
Furthermore, if the solutions of \eqref{E3.1} are bounded, then $\Psi$ satisfies a Lipschitz condition 
\begin{equation}\label{E3.5}
    \bigl\|\Psi(\tilde{\vartheta}(t)) - \Psi(\bar{\vartheta}(t))\bigr\| 
    \le \xi \bigl\|\tilde{\vartheta}(t) - \bar{\vartheta}(t)\bigr\|.
\end{equation}
The supremum norm of $\vartheta$, denoted by $\|\cdot\|_{\mathcal{I}_T}$, is defined as 
    $\|\vartheta\|_{\mathcal{I}_T} := \sup_{t\in\mathcal{I}_T} \|\vartheta(t)\|$.
It is obvious that $\mathcal{C}(\mathcal{I}_T, \mathbb{R}^2)$, equipped with the norm $\|\vartheta\|_{\mathcal{I}_T}$, forms a Banach space. 
Consequently, the problem \eqref{E3.1}--\eqref{E3.2} can be reformulated in 
$\mathcal{C}(\mathcal{I}_T, \mathbb{R}^2)$ as  
\begin{equation}\label{E3.6}
\begin{cases}
{ }^{\ABC} \Dat \bigl[ \Phi_p(\vartheta(t))\bigr] = \Psi(\vartheta(t)),\\
\vartheta(0)=\vartheta_0,
\end{cases} \quad t\in\mathcal{I}_T,
\end{equation}
where
\begin{equation*}
\Phi_p (\vartheta(t)) 
= \bigl(\Phi_p(x(t)), \Phi_p(y(t))\bigr).
\end{equation*}

\section{Existence and uniqueness of the solution}\label{S4}
In this section we will prove the existence and uniqueness of solutions of the nonlinear fractional system \eqref{E3.1}--\eqref{E3.2} using a fixed point theorem.
\begin{theorem}[Boundedness]\label{T1}
All solutions of the SI model \eqref{E3.1}--\eqref{E3.2} are bounded if
\begin{equation}\label{E4.1}
    |\mu - \eta| < \frac{2^{2-p} B(\alpha)}{1-\alpha} = \frac{\alpha + (1-\alpha)\Gamma(\alpha)}{(1-\alpha)\Gamma(\alpha)2^{p-2}},
\end{equation}
i.e.\ if $\Lambda := 2^{2-p} - \frac{(1-\alpha) |\mu - \eta|}{B(\alpha)} > 0$ holds.
\end{theorem}
\begin{proof}
Let $(x,y)$ be a solution of \eqref{E3.1}--\eqref{E3.2}. We have
\begin{equation*}
{ }^{\ABC} \Dat \bigl[ \Phi_p(x) + \Phi_p(y)\bigr] 
= { }^{\ABC} \Dat \bigl[ \Phi_p(x)\bigr]
+ { }^{\ABC} \Dat \bigl[ \Phi_p(y)\bigr]
= (\mu - \eta) N,
\end{equation*}
with $N(t) = x(t) + y(t)$. 
Applying the $\mathcal{AB}$ fractional integral \eqref{E2.3}, we get
\begin{equation*}
\bigl(\Phi_p(x) + \Phi_p(y)\bigr) - \bigl(\Phi_p(x_0) + \Phi_p( y_0)\bigr) = (\mu - \eta) { }^{\mathcal{AB}} I^{\alpha}N.
\end{equation*}
Using Tartar’s inequality \cite[Lemma 2.2.]{Okazawa2002}, we obtain
\begin{equation*}
\begin{split}
  \bigl| \Phi_p(x) + \Phi_p(y) \bigr| 
  &= \bigl| |x|^{p-2}x + |y|^{p-2}y \bigr|\\
  &= \bigl| |x|^{p-2}x - |y|^{p-2}(-y) \bigr|
  = \bigl| |x|^{p-2}x - |(-y)|^{p-2}(-y) \bigr|\\
  &\ge 2^{2-p} |x - (-y)|^{p-1} = 2^{2-p} |x + y|^{p-1} = 2^{2-p} |N|^{p-1}.
\end{split}
\end{equation*}
Since we assume positive initial data $x_0\ge0$ and $y_0\ge0$, it follows
\begin{align*}
  \bigl| \Phi_px_0) + \Phi_p(y_0) \bigr| 
  &= \bigl| |x_0|^{p-2}x_0 + |y_0|^{p-2}y_0 \bigr| 
  = |x_0|^{p-2}x_0 + |y_0|^{p-2}y_0\\
  &\le |x_0 + y_0|^{p-2}x_0 + |x_0 + y_0|^{p-2}y_0 \le |N_0|^{p-2}N_0.
\end{align*}
By the inverse triangular inequality, we have
\begin{equation*}
    2^{2-p} |N|^{p-1} - |N_0|^{p-2}N_0 \le |\mu - \eta| { }^{\mathcal{AB}} I^{\alpha} |N|.
\end{equation*}
If $N \in [-1, 1]$, then $N$ is bounded. 
Else, we get $|N| \le |N|^{p-1}$ (since $p\ge 2$). After that,
\begin{equation*}
\begin{split}
   2^{2-p} |N| - |N_0|^{p-2}N_0 &\le |\mu - \eta| { }^{\mathcal{AB}} I^{\alpha} |N|\\
   &\le \frac{(1-\alpha) |\mu - \eta|}{B(\alpha)}|N| + \frac{\alpha |\mu - \eta|}{B(\alpha)\Gamma(\alpha)}\int_0^t |N(s)| (t-s)^{\alpha-1} \,ds,
\end{split}
\end{equation*}
which gives
\begin{equation*}
\Lambda |N| \le |N_0|^{p-2}N_0 + \frac{\alpha (\mu + \eta)}{B(\alpha)\Gamma(\alpha)}\int_0^t |N(s)| (t-s)^{\alpha-1} \,ds.
\end{equation*}
Finally, Gronwall's inequality yields
\begin{equation*}
\begin{split}
   |N(t)| &\le \frac{|N_0|^{p-2}N_0}{\Lambda} 
   \exp\Bigl( \frac{\alpha (\mu + \eta)}{\Lambda B(\alpha)\Gamma(\alpha)} 
    \int_0^t (t-s)^{\alpha-1} \,ds\Bigr)\\
   &\le \frac{|N_0|^{p-2}N_0}{\Lambda} \exp\Bigl( \frac{\alpha (\mu + \eta) t^\alpha}{\Lambda B(\alpha)\Gamma(\alpha) \alpha}\Bigr)\\
   &= \frac{|N_0|^{p-2}N_0}{\Lambda} \exp\Bigl( \frac{(\mu + \eta) t^\alpha}{\Lambda B(\alpha)\Gamma(\alpha)}\Bigr).
\end{split}
\end{equation*}
This establishes the boundedness of solutions.
\end{proof}

In the rest of this work, we assume that the boundedness condition \eqref{E4.1} is satisfied.
\begin{theorem}[Positivity]\label{T2}
All solutions of the SI model \eqref{E3.1}--\eqref{E3.2} are positive for positive 
data.
\end{theorem}
\begin{proof}
Let $(x,y)$ be a solution of the SI model \eqref{E3.1}--\eqref{E3.2}.

First, suppose $x<0$. Applying \eqref{E2.4} to the first equation of \eqref{E3.1}, we have
\begin{equation*}
   \Phi_p(x) - \Phi_p(x_0) = { }^{\mathcal{AB}} I^{\alpha} (\mu N - \beta x y - ux - \eta x).
\end{equation*}
Since $x_0 \ge 0$ and by the hypothesis, we get
\begin{equation*}
    { }^{\mathcal{AB}} I^{\alpha} (\mu N - \beta x y - ux - \eta x) < 0.
\end{equation*}
Given that ${ }^{\mathcal{AB}} I^{\alpha}$ is increasing, we get
\begin{equation}\label{E4.2}
   \mu N - \beta x y  < ux + \eta x < 0,\quad\text{i.e.}\quad
\end{equation}
From \eqref{E3.1} and \eqref{E4.2}, we get
\begin{equation}\label{E4.3}
    { }^{\ABC} \Dat \bigl[\Phi_p(x)\bigr]
     = \mu N-\beta xy - ux - \eta x < - ux - \eta x .
\end{equation}
Knowing that $x<0$, so its negative part
$x^- = \max(0, -x) >0$. 
Knowing that $x$ is bounded, so by multiplying  \eqref{E4.3} by $x^-$, using \eqref{E2.7} and \eqref{E4.2}, we get
\begin{equation*}
  \frac{p-1}{p} { }^{\ABC} \Dat \bigl[ |x|^{p-2} (x^-)^2\bigr] 
  \le { }^{\ABC} \Dat \bigl[\Phi_p(x)\bigr] x^- < - (u + \eta) (x^-)^2 < 0,
\end{equation*}
i.e.\
\begin{equation*}
{ }^{\ABC} \Dat \bigl[ |x|^{p-2} (x^-)^2\bigr] < 0.
\end{equation*}
Applying the Laplace transform to this equation (cf.\ \cite{Din2021}),
we have
\begin{equation*}
  \frac{B(\alpha) s^\alpha}{\alpha + (1-\alpha)s^\alpha} 
  \mathscr{L}\bigl(|x|^{p-2} (x^-)^2\bigr) - \frac{B(\alpha) s^{\alpha-1}}{\alpha + (1-\alpha)s^\alpha} |x_0|^{p-2} (x_0^-)^2 < 0.
\end{equation*}
Since $x_0 \ge 0$, then $x_0^- = 0$. Accordingly, we obtain
\begin{equation*}
   \mathscr{L}\bigl(|x|^{p-2} (x^-)^2\bigr) < 0.
\end{equation*}
Applying the inverse Laplace transform to this inequality, we have that $|x|^{p-2} (x^-)^2 < 0$, which contradicts the hypothesis. So we get $x \ge 0$.

Next, to establish the positivity of $y$, we use the same methodology as in \cite{AltafKhan2020, Phukan2023}. 
We assume that there exists a $0 < \tau \in \mathcal{I}_T$ such that 
\begin{equation*}
    \Phi_p(y(\tau)) = 0 \ \text{ and } \ \Phi_p(y(t)) < 0, \ \ \forall t \in (\tau, \ell],
\end{equation*}
where $\ell\in\mathcal{I}_T$ is sufficiently close to $\tau$.
Using \eqref{E3.1}, we get 
\begin{equation*}
   { }^{\ABC} \Dat \bigl[ \Phi_p(y(\tau))\bigr] = u(\tau)x(\tau) \ge0.
\end{equation*}
Applying the mean value theorem for the $\ABC$ derivative \cite[Theorem 2.2.]{Fernandez2018}, we find 
\begin{equation*}
   \Phi_p(y(t)) - \Phi_p(y(\tau)) = \Phi_p(y(t)) \ge0,
\end{equation*}
which contradicts the assumption $\Phi_p(y(t)) < 0$. Consequently, $\Phi_p(y(t))\ge0$, and thus $y\ge0$.
\end{proof}

We focus our attention on establishing that the 
SI model \eqref{E3.1}--\eqref{E3.2} has a unique solution.

\begin{theorem}[Characterization of the solution]\label{T3}
$\vartheta$ is a solution of \eqref{E3.1}--\eqref{E3.2} if and only if
\begin{equation}\label{E4.4}
\begin{split}
\vartheta(t) &= \Phi_{p^*}\Bigl[ \Phi_p(\vartheta_0) + { }^{\mathcal{AB}} I^{\alpha} \bigl(\Psi(\vartheta(t))\bigr) \Bigr]\\
&= \Phi_{p^*}\Bigl[ \Phi_p(\vartheta_0) + \frac{1-\alpha}{B(\alpha)}\Psi(\vartheta(t)) +\frac{\alpha}{B(\alpha)\Gamma(\alpha)}
\int_0^t \Psi(\vartheta(s)) (t-s)^{\alpha-1} \,ds\Bigr]
=: \mathcal{T}(\vartheta(t)).  
\end{split}
\end{equation}
\end{theorem}
\begin{proof}
By applying the $\mathcal{AB}$ fractional integral to the reformulated problem \eqref{E3.6}, we obtain
\begin{equation*}
    \Phi_p(\vartheta(t)) 
     = \Phi_p(\vartheta_0) + \frac{1-\alpha}{B(\alpha)}\Psi(\vartheta(t)) 
     +\frac{\alpha}{B(\alpha)\Gamma(\alpha)}
    \int_0^t \Psi(\vartheta(s)) (t-s)^{\alpha-1} \,ds.
\end{equation*}
Consequently,
\begin{equation*}
     \vartheta(t) = \mathcal{T}(\vartheta(t)) 
     := \Phi_{p^*}\biggl[ \Phi_p(\vartheta_0) 
     + \frac{1-\alpha}{B(\alpha)}\Psi(\vartheta(t)) +\frac{\alpha}{B(\alpha)\Gamma(\alpha)}\int_0^t \Psi(\vartheta(s)) (t-s)^{\alpha-1}\,ds\biggr].
\end{equation*}
\end{proof}
To study the existence and uniqueness of the solution of the proposed problem with the fractional derivative $\ABC$, 
on the basis of the fixed point theorem and the equation \eqref{E4.4}, the iterative fixed point formula is given by
\begin{equation}\label{E4.5}
\begin{cases}
\vartheta^{n+1}(t) &= \Phi_{p^*}\Bigl[ \Phi_p(\vartheta_0)
 + \frac{1-\alpha}{B(\alpha)}\Psi(\vartheta^n(t)) +\frac{\alpha}{B(\alpha)\Gamma(\alpha)}
\int_0^t \Psi(\vartheta^n(s)) (t-s)^{\alpha-1} \,ds\Bigr],\\
\vartheta^0(t)&=\vartheta_0.
\end{cases} \quad t\in[0,T],
\end{equation}
The difference between successive iterations is considered as follows
\begin{equation}\label{E4.6}
  e^{n+1}(t) = \vartheta^{n+1}(t) - \vartheta^{n}(t)
  = \mathcal{T}(\vartheta^{n+1}(t)) - \mathcal{T}(\vartheta^{n}(t)),\quad n\ge0.
\end{equation}
This leads us to note that
\begin{equation*}
    \vartheta^{n}(t) = \sum_{k=0}^n e^k(t), \ \text{ with } \ e^0(t) = \vartheta^0(t).
\end{equation*}

\begin{lemma}\label{L4}
Assuming the Lipschitz condition \eqref{E3.5}, we have
\begin{equation*}
   \|e^{n}\|_{\mathcal{I}_T} 
   \le \biggl[\frac{(p^*-1)k^{p^*-2}\xi}{B(\alpha)}\Bigl(1-\alpha + \frac{t^{\alpha}}{\Gamma(\alpha)}\Bigr) \biggr]^n \|\vartheta^0\|_{\mathcal{I}_T}.
\end{equation*}
\end{lemma}
\begin{proof}[Proof by induction]
Let $n\in\mathbb{N}$, for $n=0$ we get
\begin{equation*}
    \| e^0(t) \|_{\mathcal{I}_T} 
    = \| \vartheta^0(t) \|_{\mathcal{I}_T} \le \biggl[\frac{(p^*-1)k^{p^*-2}\xi}{B(\alpha)}\Bigl(1-\alpha + \frac{t^{\alpha}}{\Gamma(\alpha)}\Bigr) \biggr]^0 \| \vartheta^0(t) \|_{\mathcal{I}_T}.
\end{equation*}
For $n = 1$, by considering the equations \eqref{E4.5}, \eqref{E4.6}, Lemma~\ref{L2}, and the triangular inequality of norms, we find
\begin{equation*}
\begin{split}
\| e^{1}(t) \|_{\mathcal{I}_T} 
&= \| \vartheta^{1}(t) - \vartheta^0(t) \|_{\mathcal{I}_T} \\
&= \Bigl\| \Phi_{p^*}\Bigl[\Phi_p(\vartheta_0) + { }^{\mathcal{AB}} I^{\alpha} \bigl(\Psi(\vartheta^0(t))\bigr) \Bigr] 
- \Phi_{p^*}\Bigl[ \Phi_p(\vartheta_0) + { }^{\mathcal{AB}} I^{\alpha}    
    \bigl(\Psi(\vartheta^{-1}(t))\bigr) \Bigr]\Bigr\|_{\mathcal{I}_T}\\
&\le (p^*-1) k^{p^*-2} \Bigl\| { }^{\mathcal{AB}} I^{\alpha} 
\bigl(\Psi(\vartheta^0(t))\bigr) - { }^{\mathcal{AB}} I^{\alpha} 
\bigl(\Psi(\vartheta^{-1}(t))\bigr)\Bigr\|_{\mathcal{I}_T}\\
&\le (p^*-1)k^{p^*-2}\Big\| \frac{1-\alpha}{B(\alpha)}\bigl(\Psi(\vartheta^0(t)) - \Psi(\vartheta^{-1}(t))\bigr)\\
&\qquad + \frac{\alpha}{B(\alpha) \Gamma(\alpha)} 
   \int_0^t \bigl(\Psi(\vartheta^0(s)) - \Psi(\vartheta^{-1}(s))\bigr) (t - s)^{\alpha-1} \,ds\Big\|_{\mathcal{I}_T} \\ 
&\le (p^*-1)k^{p^*-2}\frac{1-\alpha}{B(\alpha)} 
    \bigl\|\Psi(\vartheta^0(t)) - \Psi(\vartheta^{-1}(t))\bigr\|_{\mathcal{I}_T}\\  
&\qquad + (p^*-1)k^{p^*-2}\frac{\alpha}{B(\alpha) \Gamma(\alpha)} 
\int_0^t \bigl\| \Psi(\vartheta^0(s)) - \Psi(\vartheta^{-1}(s))\bigr\|_{\mathcal{I}_T} (t - s)^{\alpha-1} \,ds.
\end{split}
\end{equation*}
Afterwards, 
\begin{equation*}
\begin{split}
\| e^{1}(t) \|_{\mathcal{I}_T} &\le (p^*-1)k^{p^*-2}\frac{1-\alpha}{B(\alpha)}\xi 
\| \vartheta^0(t) - \vartheta^{-1}(t) \|_{\mathcal{I}_T}\\
&\qquad + (p^*-1)k^{p^*-2} \frac{\alpha}{B(\alpha) \Gamma(\alpha)}\xi 
    \int_0^t \| \vartheta^0(s) - \vartheta^{-1}(s) \|_{\mathcal{I}_T} (t - s)^{\alpha-1} \,ds\\
&\le (p^*-1)k^{p^*-2} \frac{1-\alpha}{B(\alpha)}\xi 
\|e^0(t)\|_{\mathcal{I}_T} + (p^*-1)k^{p^*-2} \frac{\alpha}{B(\alpha) \Gamma(\alpha)}\xi \int_0^t \|e^0(s)\|_{\mathcal{I}_T} (t - s)^{\alpha-1} \,ds\\
&\le (p^*-1)k^{p^*-2}\frac{1-\alpha}{B(\alpha)}\xi  \|\vartheta^0\|_{\mathcal{I}_T} 
+ (p^*-1)k^{p^*-2} \frac{\alpha}{B(\alpha) \Gamma(\alpha)}\xi
\|\vartheta^0\|_{\mathcal{I}_T} \int_0^t (t - s)^{\alpha-1} \,ds\\
&\le (p^*-1)k^{p^*-2} \frac{1-\alpha}{B(\alpha)}\xi \|\vartheta^0\|_{\mathcal{I}_T} 
+ (p^*-1)k^{p^*-2} \frac{\alpha}{B(\alpha) \Gamma(\alpha)\alpha}\xi 
\|\vartheta^0\|_{\mathcal{I}_T} t^{\alpha}\\
&\le \biggl[\frac{(p^*-1)k^{p^*-2}\xi}{B(\alpha)}\Bigl(1-\alpha + \frac{t^{\alpha}}{\Gamma(\alpha)}\Bigr) \biggr]^1 \|\vartheta^0\|_{\mathcal{I}_T}.
\end{split}
\end{equation*}
We then assume that the property is true for order $n$, and show that it is true for $n+1$. Using the same methodology as before, we obtain
\begin{equation*}
\begin{split}
  \|e^{n+1}(t)\|_{\mathcal{I}_T} &= \| \vartheta^{n+1}(t) - \vartheta^{n}(t) \|_{\mathcal{I}_T} \\
  &\le (p^*-1)k^{p^*-2}\frac{1-\alpha}{B(\alpha)}\xi 
  \|e^n(t)\|_{\mathcal{I}_T} 
  + (p^*-1)k^{p^*-2}\frac{\alpha}{B(\alpha) \Gamma(\alpha)}\xi 
  \int_0^t \|e^n(s)\|_{\mathcal{I}_T} (t - s)^{\alpha-1} \,ds.
\end{split}
\end{equation*}
Thus, we have
\begin{equation*}
\begin{split}
\| e^{n+1}(t) \|_{\mathcal{I}_T} 
&\le (p^*-1)k^{p^*-2}\frac{1-\alpha}{B(\alpha)}\xi 
\Biggl( \biggl[\frac{(p^*-1)k^{p^*-2}\xi}{B(\alpha)}
\Bigl(1-\alpha + \frac{t^{\alpha}}{\Gamma(\alpha)}\Bigr) \biggr]^n \|\vartheta^0\|_{\mathcal{I}_T} \Biggr)\\
&\qquad + \frac{(p^*-1)k^{p^*-2}\alpha}{B(\alpha) \Gamma(\alpha)}\xi 
\int_0^t \Biggl( \biggl[\frac{(p^*-1)k^{p^*-2}\xi}{B(\alpha)}
\Bigl(1-\alpha + \frac{t^{\alpha}}{\Gamma(\alpha)}\Bigr) \biggr]^n \|\vartheta^0\|_{\mathcal{I}_T} \Biggr) (t-s)^{\alpha-1}\,ds\\
& \le \biggl[\frac{(p^*-1)k^{p^*-2}\xi}{B(\alpha)}\Bigl(1-\alpha + \frac{t^{\alpha}}{\Gamma(\alpha)}\Bigr) \biggr]^n 
\|\vartheta^0\|_{\mathcal{I}_T} 
\bigg( (p^*-1)k^{p^*-2}\frac{1-\alpha}{B(\alpha)}\xi\\
&\qquad + (p^*-1)k^{p^*-2}\frac{\alpha}{B(\alpha) \Gamma(\alpha)}\xi \int_0^t (t - s)^{\alpha-1} \,ds \bigg)\\
& \le \biggl[\frac{(p^*-1)k^{p^*-2}\xi}{B(\alpha)}\Bigl(1-\alpha + \frac{t^{\alpha}}{\Gamma(\alpha)}\Bigr) \biggr]^{n+1} \|\vartheta^0\|_{\mathcal{I}_T}.
\end{split}
\end{equation*}
\end{proof}

\begin{theorem}[Uniqueness]\label{T4}
If there is a $t_0$ which satisfies
\begin{equation}\label{E4.7}
    \frac{(p^*-1)k^{p^*-2}\xi}{B(\alpha)}\Bigl(1-\alpha + \frac{t_0^{\alpha}}{\Gamma(\alpha)}\Bigr) < 1,
\end{equation}
then the problem \eqref{E3.6} admits one and only one solution.
\end{theorem}
\begin{proof}
Recall from \eqref{E4.6} that $\vartheta^{n} = \sum_{j=0}^n e^j$. Then we have,
\begin{equation*}
   \| \vartheta^{n}\|_{\mathcal{I}_T} 
   \le \| \vartheta^0\|_{\mathcal{I}_T} \sum_{j=0}^n \biggl[ \frac{(p^*-1)k^{p^*-2}\xi}{B(\alpha)}\Bigl(1-\alpha + \frac{t^{\alpha}}{\Gamma(\alpha)}\Bigr) \biggr]^{j} .
\end{equation*}
For $t = t_0$, the above relation reads
\begin{equation}\label{E4.8}
    \| \vartheta^{n}\|_{\mathcal{I}_T} \le \| \vartheta^0\|_{\mathcal{I}_T} \sum_{j=0}^n \biggl[ \frac{(p^*-1)k^{p^*-2}\xi}{B(\alpha)}
    \Bigl(1-\alpha + \frac{t_0^{\alpha}}{\Gamma(\alpha)}\Bigr) \biggr]^{j} .
\end{equation}
Since we have 
\begin{equation*}
\frac{(p^*-1)k^{p^*-2}\xi}{B(\alpha)}\Bigl(1-\alpha + \frac{t_0^{\alpha}}{\Gamma(\alpha)}\Bigr) < 1,
\end{equation*}
the geometric series 
\begin{equation*}
     \sum_{j=0}^n \biggl[ \frac{(p^*-1)k^{p^*-2}\xi}{B(\alpha)}\Bigl(1-\alpha + \frac{t_0^{\alpha}}{\Gamma(\alpha)}\Bigr) \biggr]^j
\end{equation*}
is convergent. 
Consequently, the series $(\vartheta^n)$ exists and is bounded for any $n\in\mathbb{N}$, and we have
\begin{equation*}
    \lim_{n\to\infty}\| \vartheta^{n}\|_{\mathcal{I}_T} < \infty .
\end{equation*}
We also consider the relation
\begin{equation*}
    \mathcal{R}^n = \vartheta - \vartheta^n.
\end{equation*}
Using the same technique as in the proof of Lemma~\ref{L3}, we can show the following inequality
\begin{equation*}
\begin{split}
    \| \mathcal{R}^n \|_{\mathcal{I}_T} 
    &\le \frac{(p^*-1)k^{p^*-2}\xi}{B(\alpha)}
    \Bigl(1-\alpha + \frac{t^{\alpha}}{\Gamma(\alpha)}\Bigr)
    \|\mathcal{R}^{n-1} \|_{\mathcal{I}_T}\\
&\le \biggl[\frac{(p^*-1)k^{p^*-2}\xi}{B(\alpha)}
    \Bigl(1-\alpha + \frac{t^{\alpha}}{\Gamma(\alpha)}\Bigr) \biggr]^n \|
    \mathcal{R}^0 \|_{\mathcal{I}_T}.
\end{split}
\end{equation*}
Now, taking the limit $n\to\infty$ in the above relation we find 
\begin{equation*}
    \lim_{n\to\infty}\| \vartheta - \vartheta^{n}\|_{\mathcal{I}_T} = 0, \ \text{ i.e. } \ \lim_{n\to\infty}\vartheta^{n} = \vartheta.
\end{equation*}
Therefore, the existence of a solution is proven.

Now we need to prove uniqueness. 
Let us assume that there are two solutions to the reformulated system \eqref{E3.6}, namely by $\tilde{\vartheta}$ and $\bar{\vartheta}$, then
\begin{equation*}
\begin{split}
  \| \tilde{\vartheta}(t) - \bar{\vartheta}(t) \|_{\mathcal{I}_T} 
  &\le (p^*-1)k^{p^*-2}\Big\| \frac{1-\alpha}{B(\alpha)}
  \bigl(\Psi(\tilde{\vartheta}(t)) - \Psi(\bar{\vartheta}(t))\bigr)\\
   &\qquad + \frac{\alpha}{B(\alpha) \Gamma(\alpha)} \int_0^t \bigl(\Psi(\tilde{\vartheta}(s)) - \Psi(\bar{\vartheta}(s))\bigr) (t - s)^{\alpha-1} \,ds\Big\|_{\mathcal{I}_T} \\ 
   &\le (p^*-1)k^{p^*-2}\frac{1-\alpha}{B(\alpha)} 
\bigl\|\Psi(\tilde{\vartheta}(t)) - \Psi(\bar{\vartheta}(t))\bigr\|_{\mathcal{I}_T}\\  
  &\qquad + (p^*-1)k^{p^*-2}\frac{\alpha}{B(\alpha) \Gamma(\alpha)} 
\int_0^t \bigl\| \Psi(\tilde{\vartheta}(s)) - \Psi(\bar{\vartheta}(s))\bigr\|_{\mathcal{I}_T} (t - s)^{\alpha-1} \,ds\\
   &\le (p^*-1)k^{p^*-2}\frac{1-\alpha}{B(\alpha)}\xi 
\bigl\|\tilde{\vartheta}(t) - \bar{\vartheta}(t)\bigr\|_{\mathcal{I}_T}\\
  &\qquad + (p^*-1)k^{p^*-2}\frac{\alpha}{B(\alpha) \Gamma(\alpha)}\xi 
\int_0^t \bigl\| \tilde{\vartheta}(s) - \bar{\vartheta}(s)\bigr\|_{\mathcal{I}_T} (t - s)^{\alpha-1} \,ds\\
  &\le (p^*-1)k^{p^*-2}\frac{1-\alpha}{B(\alpha)}\xi \|\tilde{\vartheta} - \bar{\vartheta}\|_{\mathcal{I}_T} + (p^*-1)k^{p^*-2}\frac{t^{\alpha}}{B(\alpha) \Gamma(\alpha)}\xi \|\tilde{\vartheta} - \bar{\vartheta}\|_{\mathcal{I}_T} .
\end{split}
\end{equation*}
This estimate shows that
\begin{equation*}
    \| \tilde{\vartheta} - \bar{\vartheta} \|_{\mathcal{I}_T} \le 0,\quad
    \text{i.e.}\quad
    \tilde{\vartheta} = \bar{\vartheta}.
\end{equation*}
\end{proof}

\section{Stability analysis}\label{S5}
In this section we give an analysis of the Hyers-Ulam stability for the reformulated problem \eqref{E3.6}.
The study of stability of functional equations originated from an open question in 1964 (cf.\ \cite{Forti1995, Ulam1964}), which focused on the stability of a group homomorphism. 
Given a group $\mathcal{G}$ and a metric group $(\mathcal{G}^\prime, d)$. 
Let $\varepsilon > 0$, is there a $a > 0$ such that, if $f\colon\mathcal{G} \to \mathcal{G}^\prime$ satisfies the condition 
\begin{equation*}
   d\bigl(f(xy), f(x)f(y)\bigr) < a, \quad \text{for all } x, y \in \mathcal{G},
\end{equation*}
then there exists a homomorphism $g\colon\mathcal{G} \to \mathcal{G}^\prime$ such that 
\begin{equation*}
   d\bigl(f(x), g(x)\bigr) < \varepsilon, \quad\text{for all } x \in \mathcal{G}.
\end{equation*}
Essentially, if we have an almost homomorphism, then there is a very near homomorphism with a small error.
Next, we define the stability notion mentioned above and prove the stability of the reformulated system \eqref{E3.6}.

\begin{definition}[Hyers-Ulam stability]\label{D4}
The $\ABC$-FDE problem \eqref{E3.6} is Hyers-Ulam stable if, for any given $\mathscr{S} > 0$,
there exists a constant $\delta > 0$ such that if
\begin{equation}\label{E5.1}
    \biggl| \vartheta_i(t) - \Phi_{p^*}\Bigl[ \Phi_p(\vartheta_0^i) + \frac{1-\alpha}{B(\alpha)}\Psi(\vartheta_i(t)) +\frac{\alpha}{B(\alpha)\Gamma(\alpha)}
    \int_0^t \Psi(\vartheta_i(s)) (t-s)^{\alpha-1} \,ds\Bigr] \biggr| \le \mathscr{S},
\end{equation}
there exists $\varphi(t) = (\varphi_1(t),\varphi_2(t))$ satisfying
\begin{equation}\label{E5.2}
   \varphi_i(t) = \Phi_{p^*}\Bigl[ \Phi_p(\vartheta_0^i) +
   \frac{1-\alpha}{B(\alpha)}\Psi(\varphi_i(t)) +\frac{\alpha}{B(\alpha)\Gamma(\alpha)}
   \int_0^t \Psi(\varphi_i(s)) (t-s)^{\alpha-1} \,ds\Bigr],
\end{equation}
such that
\begin{equation*}
    |\vartheta_i(t) - \varphi_i(t)| \le \delta \mathscr{S}.
\end{equation*}
\end{definition}

The following theorem shows the stability of the problem in the sense defined above.
\begin{theorem}\label{T5}
The reformulated problem \eqref{E3.6} is Hyers-Ulam stable.
\end{theorem}
\begin{proof}
Let $\vartheta(t) = (\vartheta_1(t),\vartheta_2(t)) = (x(t), y(t))$ be a solution of \eqref{E3.6} 
and $\varphi(t) = (\varphi_1(t),\varphi_2(t))$ be an approximate solution and satisfying \eqref{E4.4}. 
Then, for $i = 1, 2,$ we have
\begin{equation*}
\begin{split}
 \bigl|\vartheta_i(t) - \varphi_i(t)\bigr| 
 &= \bigg| \Phi_{p^*}\Bigl[ \Phi_p(\vartheta_0^i) + \frac{1-\alpha}{B(\alpha)}\Psi(\vartheta_i(t)) 
 +\frac{\alpha}{B(\alpha)\Gamma(\alpha)}
 \int_0^t  \Psi(\vartheta_i(s)) (t-s)^{\alpha-1} \,ds\Bigr]\\
&\qquad - \Phi_{p^*} \Bigl[ \Phi_p(\vartheta_0^i) + \frac{1-\alpha}{B(\alpha)}\Psi(\varphi_i(t)) 
+\frac{\alpha}{B(\alpha)\Gamma(\alpha)}
\int_0^t \Psi(\varphi_i(s)) (t-s)^{\alpha-1} \,ds\Bigr]\bigg|.
\end{split}
\end{equation*}
Using Lemma~\ref{L2} leads to
\begin{equation*}
\begin{split}
\bigl|\vartheta_i(t) - \varphi_i(t)\bigr| 
&\le (p-1) k^{p^*-2} \Big|\frac{1-\alpha}{B(\alpha)}\Psi(\vartheta_i(t))
+\frac{\alpha}{B(\alpha)\Gamma(\alpha)}\int_0^t \Psi(\vartheta_i(s)) (t-s)^{\alpha-1} \,ds \\
&\qquad - \frac{1-\alpha}{B(\alpha)}\Psi(\varphi_i(t)) - \frac{\alpha}{B(\alpha)\Gamma(\alpha)}\int_0^t \Psi(\varphi_i(s)) (t-s)^{\alpha-1} \,ds \Big|\\
&\le (p-1) k^{p^*-2} \frac{1-\alpha}{B(\alpha)} \bigl|\Psi(\vartheta_i(t)) - \Psi(\varphi_i(t))\bigr|\\
&\qquad + (p-1) k^{p^*-2}\frac{\alpha}{B(\alpha)\Gamma(\alpha)}
\int_0^t \bigl|\Psi(\vartheta_i(s)) - \Psi(\varphi_i(s))\bigr| (t-s)^{\alpha-1} \,ds.
\end{split}
\end{equation*}
Afterwards,
\begin{equation*}
\begin{split}
\bigl|\vartheta_i(t) - \varphi_i(t)\bigr| 
&\le (p-1) k^{p^*-2}\xi \frac{1-\alpha}{B(\alpha)} 
  \bigl\|\vartheta_i - \varphi_i\bigr\|_{\mathcal{I}_T}\\
&\qquad + (p-1) k^{p^*-2} \xi \frac{\alpha}{B(\alpha)\Gamma(\alpha)}
\int_0^t \bigl\|\vartheta_i - \varphi_i\bigr\|_{\mathcal{I}_T} (t-s)^{\alpha-1} \,ds.
\end{split}
\end{equation*}
Then,
\begin{equation*}
\bigl|\vartheta_i(t) - \varphi_i(t)\bigr| 
\le \Bigl[(p-1) k^{p^*-2}\xi \frac{1-\alpha}{B(\alpha)} + (p-1) k^{p^*-2} \xi \frac{t^\alpha}{B(\alpha)\Gamma(\alpha)} \Bigr] \varepsilon 
\le \delta \mathscr{S},
\end{equation*}
where $\delta = \frac{(p-1) k^{p^*-2}}{B(\alpha)}\Bigl[1-\alpha +  \frac{T^\alpha}{\Gamma(\alpha)} \Bigr] \xi$.
\end{proof}

After having shown the stability of the problem, we now want to analyze the optimal control.

\section{Optimal control analysis}\label{S6}
In this section, we will discuss the optimality analysis for the SI model \eqref{E3.1}--\eqref{E3.2}. 
We define our objective functional as follows
\begin{equation}\label{E6.1}
    \mathscr{T}(u) = \int_0^T \Bigl( \mathscr{G} I^2(t) + \frac{\varrho}{2} u^2(t)\Bigr)\,dt,
\end{equation}
where $\varrho$ is a weight constant for the vaccination rate $u\in \mathcal{U}_{ad}$, 
while $\mathscr{G}$ are the proportional weights assigned to $S$ and $I$, respectively.
In the following, we will first show the existence of an optimal solution and then investigate the optimality conditions.

\subsection{Existence of an optimal solution}\label{SubS6.1}
Using minimizing sequences, we prove the existence of optimal control. 
The proof is based on the following lemma, which we introduce below.

\begin{lemma}[\protect{\cite[Proposition 2.1/Corollary 2.1]{Djida2018}}]\label{L5}
Let $\nu, \varphi \in C^{\infty}(\mathcal{I}_T)$. Then
\begin{equation}\label{E6.2}
\begin{split}
\int_0^T ({ }^{\ABC} \Dat \nu) \ \varphi \,dt 
&= -\int_0^T \nu \ ({ }^{\ABC}_T \Dat \varphi)\,dt - \frac{B(\alpha)}{1-\alpha} \nu(0) \int_0^T \varphi(t) E_{\alpha, \alpha}[-\gamma t^{\alpha}]\,dt\\
&\quad + \frac{B(\alpha)}{1-\alpha} \varphi(T) 
  \int_0^T \nu(t) E_{\alpha, \alpha}[-\gamma (T-t)^{\alpha}] \,dt.
\end{split}
\end{equation}
\end{lemma}

\begin{theorem}\label{T6}
The reformulated system \eqref{E3.6} admits an optimal solution $\vartheta^* = \vartheta(u^*)\in L^\infty(\mathcal{I}_T)$ that minimizes the objective functional \eqref{E6.1}.
\end{theorem}
\begin{proof}
Let $((\vartheta^n, u^n))_n$ such as
\begin{equation*}
   \mathscr{T}(u^*) \le\mathscr{T}(u^n) \le\mathscr{T}(u^*) + \frac{1}{n}, \quad \forall n \in \mathbb{N}^*,
\end{equation*}
where $u^n\in U_{ad}$ and $\vartheta^n = \vartheta(u^n)$ satisfying the system
\begin{equation*}
\begin{cases}
   { }^{\ABC} \Dat \bigl[\Phi_p(\vartheta^n)\bigr] 
   = \Psi(\vartheta^n), &\text{ in } \mathcal{I}_T,\\ 
   \vartheta^n(0) = \vartheta_0.&
\end{cases}
\end{equation*}
Based on Theorem~\ref{T1}, we can say that $\vartheta^n$ is bounded.
Consequently, the second term $\Psi(\vartheta^n)$ is also bounded. 
Thus, there exists a positive constant $c$ such that
\begin{equation*}
    \Bigl\| { }^{\ABC} \Dat \bigl[\Phi_p(\vartheta^n)\bigr]\Bigr\|_2 \le c.
\end{equation*}
Then there exists a subsequence of ($\vartheta^n$), again denoted by ($\vartheta^n$), such that
\begin{gather}
 { }^{\ABC} \Dat \bigl[\Phi_p(\vartheta^n)\bigr] 
 \rightharpoonup  \phi \text { weakly in } \bigl(L^2(\mathcal{I}_T)\bigr)^2, \notag\\
 \vartheta^n \rightharpoonup  \vartheta^* \text { weakly-star in } 
  \bigl(L^\infty(\mathcal{I}_T)\bigr)^2,\label{E6.3}\\
 \vartheta^n \rightharpoonup  \vartheta^* \text { weakly in } \bigl(L^2(\mathcal{I}_T)\bigr)^2.\notag
\end{gather}
It's worth noting that $D^\prime(\mathcal{I}_T)$ is the dual of $C_0^\infty(\mathcal{I}_T)$.
Then, for all $\varphi\in C_0^\infty(\mathcal{I}_T)$, we get
\begin{align*}
\int_0^T \Phi_p(x^n) \ \bigl({ }^{\ABC}_T \Dat \varphi\bigr) \,dt &\longrightarrow 
\int_0^T \Phi_p(x^*) \ \bigl({ }^{\ABC}_T \Dat \varphi\bigr) \,dt,\\
\varphi(T) \int_0^T \Phi_p(x^n) E_{\alpha, \alpha}[-\gamma(T-t)^{\alpha}] \,dt &\longrightarrow \varphi(T) \int_0^T \Phi_p(x^*) E_{\alpha, \alpha}[-\gamma(T-t)^{\alpha}]\, dt,
\end{align*}
and
\begin{align*}
\int_0^T \Phi_p(y^n) \ \bigl({ }^{\ABC}_T \Dat \varphi\bigr) \,dt 
&\longrightarrow 
\int_0^T \Phi_p(y^*) \ \bigl({ }^{\ABC}_T \Dat \varphi\bigr) \,dt,\\
\varphi(T) \int_0^T \Phi_p(y^n) E_{\alpha,\alpha} [-\gamma(T-t)^{\alpha}]\, dt &\longrightarrow \varphi(T) 
\int_0^T \Phi_p(y^*) E_{\alpha, \alpha}[-\gamma(T-t)^{\alpha}] \,dt.
\end{align*}
By Lemma \ref{L5}, we obtain
\begin{align*}
& { }^{\ABC} \Dat \bigl[\Phi_p(x^n)\bigr] \rightharpoonup { }^{\ABC} \Dat \bigl[\Phi_p(x^*)\bigr] \text{ weakly in } D^\prime(\mathcal{I}_T),\\
& { }^{\ABC} \Dat \bigl[\Phi_p(y^n)\bigr] \rightharpoonup { }^{\ABC} \Dat \bigl[\Phi_p(y^*)\bigr] \text{ weakly in } D^\prime(\mathcal{I}_T).
\end{align*}
Expressing $x^n y^n - x^* y^* = (x^n-x^*) y^n + x^* (y^n - y^*)$, based on \eqref{E6.3} and the boundedness of $(x^n)$ and $(y^n)$, 
it follows that $x^n y^n \to x^* y^*$ in $L^2(\mathcal{I}_T)$. 
In addition, we observe $u^n \to u^*$ in $L^2(\mathcal{I}_T)$ along a subsequence of $(u^n)$, denoted again by $(u^n)$. 
Using the closeness and convexity of $U_{a d}$, it can be deduced that $U_{a d}$ is weakly closed. 
Consequently, $u^* \in U_{a d}$, and similarly, as described earlier, $u^n x^n \to u^* x^*$ in $L^2(\mathcal{I}_T)$.

Continuing, we can take the limit in the system satisfied by $\vartheta^n$ as $n \to\infty$, 
which leads to the conclusion that the optimal solution for \eqref{E3.1}--\eqref{E3.2} is given by $(\vartheta^*, u^*)$.
\end{proof}


\subsection{Optimality conditions}\label{SubS6.2}
Let $\vartheta^{\varepsilon} = (x^{\varepsilon}, y^{\varepsilon}) = (x,y)(u^{\varepsilon})$ and $\vartheta^* = (x^*, y^*) = (x, y)(u^*)$ be the approximate and optimal solutions of \eqref{E3.1}--\eqref{E3.2} and \eqref{E6.1} respectively, 
where $u^{\varepsilon}=u^* + \varepsilon u \in U_{a d}, \ \forall u \in U_{a d}$. 
We subtract the system associated with $\vartheta^*$ from the one corresponding to $\vartheta^{\varepsilon} = \vartheta^* + \varepsilon \mathscr{Z}^{\varepsilon}$ (with $\mathscr{Z}^{\varepsilon} = ( \mathscr{Z}_1^{\varepsilon}, \mathscr{Z}_2^{\varepsilon})$), we get
\begin{equation}\label{E6.4}
\begin{cases}
{ }^{\ABC} \Dat\bigl[\Phi_p(\vartheta^{\varepsilon}) - \Phi_p(\vartheta^*)\bigr] 
=  \Psi(\vartheta^{\varepsilon}) - \Psi(\vartheta^*), &\text { in } \mathcal{I}_T, \\
\mathscr{Z}^{\varepsilon}(0)=0.
\end{cases}
\end{equation}
Since $\Phi_p(\vartheta^{\varepsilon}(0)) - \Phi_p(\vartheta^*(0)) = 0$, then applying \eqref{E2.4} to the first equation of \eqref{E6.4}, we have
\begin{equation*}
\Phi_p(\vartheta^{\varepsilon}) - \Phi_p(\vartheta^*) = { }^{\mathcal{AB}} I^{\alpha} \bigl( \Psi(\vartheta^{\varepsilon}) - \Psi(\vartheta^*)\bigr).
\end{equation*}
Remember that $\Phi_p$ is a regular function.
Using the mean value theorem, we find that there exists a constant $\zeta$ such that
\begin{align*}
\Phi_p(\vartheta^{\varepsilon}) - \Phi_p(\vartheta^*) &= \Phi_p^\prime(\zeta) (\vartheta^{\varepsilon} - \vartheta^*) = \varepsilon \Phi_p^\prime(\zeta) \mathscr{Z}^{\varepsilon}\\
&= \varepsilon \Phi_p^\prime(\zeta) \mathscr{Z}^{\varepsilon}(t) - \varepsilon \Phi_p^\prime(\zeta) \mathscr{Z}^{\varepsilon}(0)\\
&= { }^{\mathcal{AB}} I^{\alpha} 
\bigl( \varepsilon \Phi_p^\prime(\zeta) { }^{\ABC} \Dat [ \mathscr{Z}^{\varepsilon}(t)]\bigr).
\end{align*}
Thus,
\begin{equation*}
  \Phi_p^\prime(\zeta) { }^{\ABC} \Dat \bigl[ \mathscr{Z}^{\varepsilon}(t)\bigr]
  = \frac{\Psi(\vartheta^{\varepsilon}) - \Psi(\vartheta^*)}{\varepsilon}.
\end{equation*}
On the other side
\begin{equation*}
\frac{\Psi(\vartheta^{\varepsilon}) - \Psi(\vartheta^*)}{\varepsilon}
= \mathscr{N}^{\varepsilon} \mathscr{Z}^{\varepsilon} + \mathscr{Y}u,
\end{equation*}
with 
\begin{equation*}
\mathscr{N}^{\varepsilon} = \begin{pmatrix}
\mu - \beta y^{\varepsilon} - u^{\varepsilon} - \eta & \mu - \beta x^*\\
\beta y^{\varepsilon} + u^{\varepsilon} & \beta x^* - \eta
\end{pmatrix}
\ \text{ and } \
\mathscr{Y} = \begin{pmatrix}
-x^* \\
x^*
\end{pmatrix}.
\end{equation*}
The components of the matrix $\mathscr{N}^{\varepsilon}$ are uniformly bounded with respect to $\varepsilon$. 
By using Lemma~\ref{L5} and a similar approach described in Subsection~\ref{SubS6.1}, as $\varepsilon$ tends to 0 in \eqref{E6.4}, we obtain
\begin{equation}\label{E6.5}
\begin{cases}
\Phi_p^\prime(\vartheta^*) { }^{\ABC} \Dat \bigl[ \mathscr{Z}(t)\bigr] = \mathscr{N}\mathscr{Z} + \mathscr{Y}u, &\text { in } \mathcal{I}_T, \\ 
\mathscr{Z}(0) = 0,
\end{cases}
\end{equation}
where 
\begin{equation*}
   \mathscr{N} = \begin{pmatrix}
    \mu - \beta y^* - u^* - \eta & \mu - \beta x^*\\
    \beta y^* + u^* & \beta x^* - \eta \end{pmatrix}.
\end{equation*}
For $p = 2$, according to Theorem~\ref{T4}, we can state that the problem \eqref{E6.5} has a unique solution.
We now introduce $\mathscr{P} = (\mathscr{P}_1, \mathscr{P}_2)$ in such a way that
\begin{equation*}
    \Phi_p^\prime(\vartheta^*) \int_0^T \bigl( { }^{\ABC} \Dat[\mathscr{Z}(t)]\bigr) \mathscr{P} \,dt 
   = \int_0^T (\mathscr{N}\mathscr{Z} + \mathscr{Y}u) \mathscr{P}\,dt.
\end{equation*}
By Lemma~\ref{L5}, we have
\begin{equation*}
    \int_0^T ({ }^{\ABC} \Dat \mathscr{Z}) \ \mathscr{P} \,dt 
    = -\int_0^T \mathscr{Z} \ ({ }^{\ABC}_T \Dat \mathscr{P}) \,dt 
       + \frac{B(\alpha)}{1-\alpha} \mathscr{P}(T) 
    \int_0^T \mathscr{Z}(t) E_{\alpha, \alpha}[-\gamma (T-t)^{\alpha}] \,dt.
\end{equation*}
Consequently, the dual system corresponding to \eqref{E3.1}--\eqref{E3.2} can be formulated as
\begin{equation}\label{E6.6}
\begin{cases}
-\Phi_p^\prime(\vartheta^*){ }^{\ABC}_T \Dat \mathscr{P} - \mathscr{N} \mathscr{P} = \mathscr{W}^*\mathscr{W}\vartheta^* , &\text { in } \mathcal{I}_T, \\ 
\mathscr{P}(T) = \mathscr{W}^*\mathscr{W}\vartheta^*(T),
\end{cases}
\end{equation}
with
\begin{equation*}
\mathscr{W} = \begin{pmatrix}0 & 0 \\0 & 1\end{pmatrix}.
\end{equation*}
We can show that the problem \eqref{E6.6} has a unique solution (Theorem~\ref{T4} with $p = 2$).

\begin{theorem}\label{T7}
Let $\mathscr{P}$ be the solution of \eqref{E6.6}. Then,
\begin{equation}\label{E6.7}
    u^* = \max \biggl\{0, \min \Bigl(\frac{\mathscr{P}_1^* - \mathscr{P}_2^*}{\varrho} S^*, 1\Bigr)\biggr\}.
\end{equation}
\end{theorem}
\begin{proof}
Let $(S, I)$ be the solution of \eqref{E3.1}--\eqref{E3.2}. The Hamiltonian associated with \eqref{E6.1} is defined by
\begin{equation}\label{E6.8}
\begin{split}
\mathscr{H}(S, I, \mathscr{P}_1, \mathscr{P}_1, u) 
&=  \ \mathscr{G} I^2(t) + \frac{\varrho}{2} u^2(t)\\
&\qquad + \mathscr{P}_1 (\mu N-\beta S I - uS - \eta S) + \mathscr{P}_2 (\beta S I + uS - \eta I),
\end{split}
\end{equation}
where $(\mathscr{P}_1, \mathscr{P}_2)$ is the adjoint variable. 
Following \cite{Khan2021}, we obtain the necessary optimality conditions for \eqref{E3.1}--\eqref{E3.2} and \eqref{E6.1} such that
\begin{equation}\label{E6.9}
\begin{split}
& { }^{\ABC} \Dat \bigl[\Phi_p(S(t))\bigr] = \frac{\partial \mathscr{H}}{\partial \mathscr{P}_1}(S, I, \mathscr{P}_1, \mathscr{P}_1, u), \\
& { }^{\ABC} \Dat \bigl[ \Phi_p(I(t))\bigr] = \frac{\partial \mathscr{H}}{\partial \mathscr{P}_2}(S, I, \mathscr{P}_1, \mathscr{P}_1, u), \\ 
& \frac{\partial \mathscr{H}}{\partial u}(S^*, I^*, \mathscr{P}_1^*, \mathscr{P}_1^*, u^*) = 0.
\end{split}
\end{equation}
Since $\frac{\partial \mathscr{H}}{\partial u}(S^*, I^*, \mathscr{P}_1^*, \mathscr{P}_1^*, u^*) = 0$, then
\begin{equation*}
    u^* = \max \biggl\{0, \min \Bigl(\frac{\mathscr{P}_1^* - \mathscr{P}_2^*}{\varrho} S^*, 1\Bigr)\biggr\}.
\end{equation*}
\end{proof}

\section{Numerical simulations}\label{S7}
In this section, we present a numerical approach to solving the optimality system over the entire time interval $\mathcal{I}_T$. 
Based on \cite{Gatto2021, Lin2010, SidiAmmi2023}, we provide specific (positive) parameters and initial conditions, as detailed in Table~\ref{Tab2}, assuming the boundedness constraint \eqref{E4.1}.

\setlength{\tabcolsep}{0.5cm}
\begin{table}[H]
\centering
\caption{Initial conditions and parameters values.}\label{Tab2}
\begin{tabular}{|c||c||c|}
\hline Symbol & Description & Value\\
\hline\hline 
$S_0$ & Initial susceptible individuals & 10\\
\hline $I_0$ & Initial infected individuals & 4\\
\hline $\mu$ & Birth rate & 0.04\\
\hline $\beta$ & Effective contact rate & 0.09\\
\hline $\eta$ & Natural mortality rate & 0.04\\
\hline $T$ & Final time & 200\\
\hline
\end{tabular}
\end{table}

\subsection{Forward-Backward Sweep Method}
We show numerical simulations related to the above optimal control problem. 
A MatLab code has been developed and several simulations have been performed using different data. 
The optimality system is solved based on an iterative discrete scheme that converges after an appropriate test similar to the one related to the \textit{Forward-Backward Sweep Method} (FBSM). 
Due to transversality conditions, the state model is solved forward in time, followed by the backward solution of the adjoint system. 
Then, the optimal control values are updated using the state and adjoint variables obtained in the previous steps. 
This iterative process is repeated until a tolerance criterion is satisfied. 
The algorithm is summarized in Algorithm~\ref{Alg1}.

\begin{property}[\cite{Djida2018}]\label{P1}
Let $\mathscr{X}\colon\mathcal{I}_T \to\mathbb{R}$. Then,
\begin{equation}\label{E7.1}
     {}^{\ABC}_T \Dat \mathscr{X}(t) = {}^{\ABC} \Dat \mathscr{X}(T-t).
\end{equation}
\end{property}

\begin{algorithm}[htb]
\caption{FBSM.}\label{Alg1}
\begin{algorithmic}[1]
\REQUIRE Let $\mathcal{S} = \{t_i = 1 + i\delta_t / i=0,\cdots \mathcal{N} - 1\}$ be a uniform subdivision on $\mathcal{I}_T$,
where $\delta_t$ is the step size of this subdivision and $\mathcal{N}$ is the number of sub-intervals. 
$\delta = 10^{-3}$ is tolerance and $Rel\_Err = -1$ is the relative error.

\STATE Variables initialization: For the state problem $(S_{\rm old}, I_{\rm old})$,
the adjoint system $(\mathscr{P}_{rm old_{1}}, \mathscr{P}_{\rm old_{2}})$, and the control $u_{\rm old}$.

\WHILE{$Rel\_Err < 0$}
\STATE Solve \eqref{E3.1} with initial condition \eqref{E3.2} given in Table \ref{Tab2}.

\STATE Under \eqref{E7.1}, solve the \eqref{E6.6} for $(\mathscr{P}_{\rm old_{1}}, \mathscr{P}_{\rm old_{2}})$ using the transversality conditions $(\mathscr{P}_{\rm old_{1}}(T), \mathscr{P}_{\rm old_{2}}(T))$ and $(S, I)$.

\STATE Use \eqref{E6.7} to update the control $u$.

\STATE Calculate relative errors: 
\begin{itemize}
\item[] $\psi_{1} = \delta\|S\| - \|S - S_{\rm old}\|$, 
\quad $\psi_{2} = \delta\|I\| - \|I - I_{\rm old}\|$,
\item[] $\psi_{3} = \delta\|\mathscr{P}_{1}\| - \|\mathscr{P}_{1} - \mathscr{P}_{\rm old_{1}}\|$, 
\quad $\psi_{4} = \delta\|\mathscr{P}_{2}\| - \|\mathscr{P}_{2} - \mathscr{P}_{\rm old_{1}}\|$,
\item[] $\psi_{5} = \delta\|u\| - \|u - u_{\rm old}\|$.
\end{itemize}

\STATE $Rel\_Err = \min\{\psi_1, \psi_2, \psi_3, \psi_4, \psi_5\}$.
\ENDWHILE
\end{algorithmic}
\end{algorithm}

\subsection{Numerical approximations in the absence of vaccination}
The Figures \ref{F2}, \ref{F3} and \ref{F4} show 
\textit{numerical approximations in the absence of vaccination} (NAAV), considering different values of $\alpha$ and $p$. 
Throughout the observation period, the number of susceptible individuals decreases as the number of infected individuals rapidly increases.
In each of these three cases ($\alpha = 1$, $\alpha = 0.9$, and $\alpha = 0.8$), we can observe that the delayed spread of the disease is directly related to the values of $p$ and $\alpha$.
Thus, whenever $p$ is large, it results in the slow spread of the disease. 
In addition, the values of $\alpha$ are also related to how quickly the disease spreads, 
because natural derivative values result in the rapid spread of the disease, while fractional values result in the slow spread of the disease.

\begin{figure}[H]
\centering
\includegraphics[height=8cm]{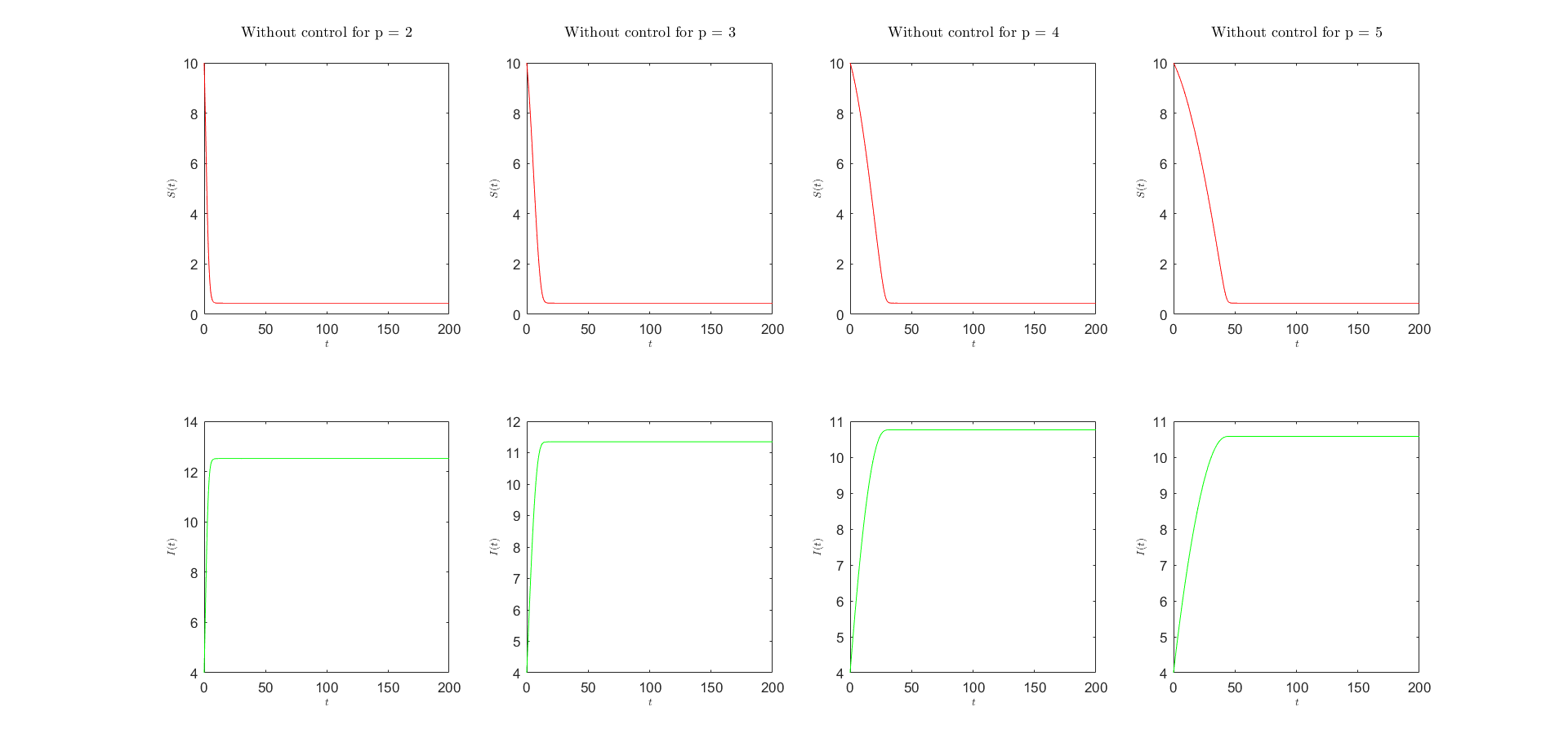}
\caption{Numerical simulations of the SI model \eqref{E3.1}--\eqref{E3.2} for $\alpha = 1$ without control.}\label{F2}
\end{figure}

\begin{figure}[htb]
\centering
\includegraphics[height=8cm]{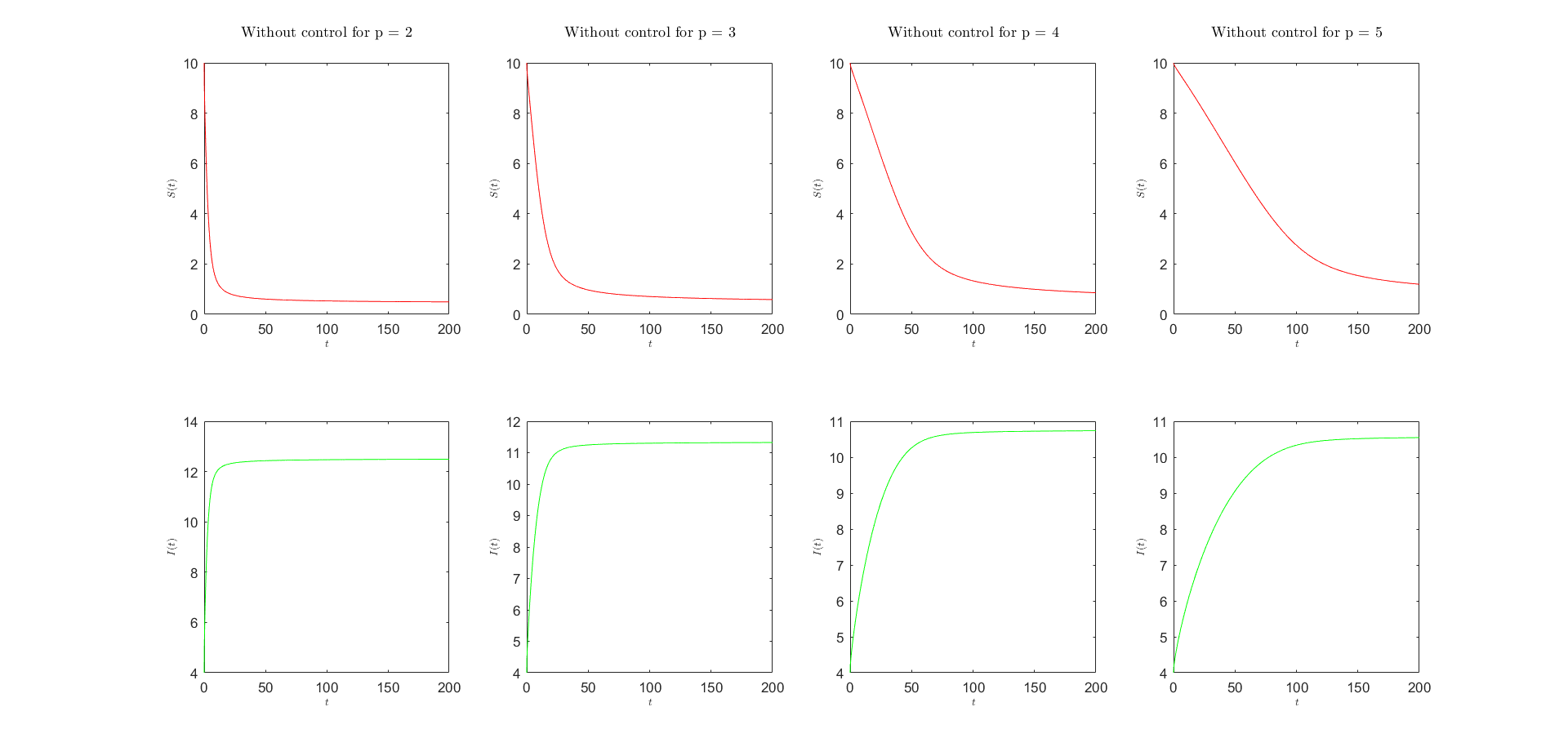}
\caption{Numerical simulations of the SI model \eqref{E3.1}--\eqref{E3.2} for $\alpha = 0.9$ without control.}\label{F3}
\end{figure}

\begin{figure}[htb]
\centering
\includegraphics[height=8cm]{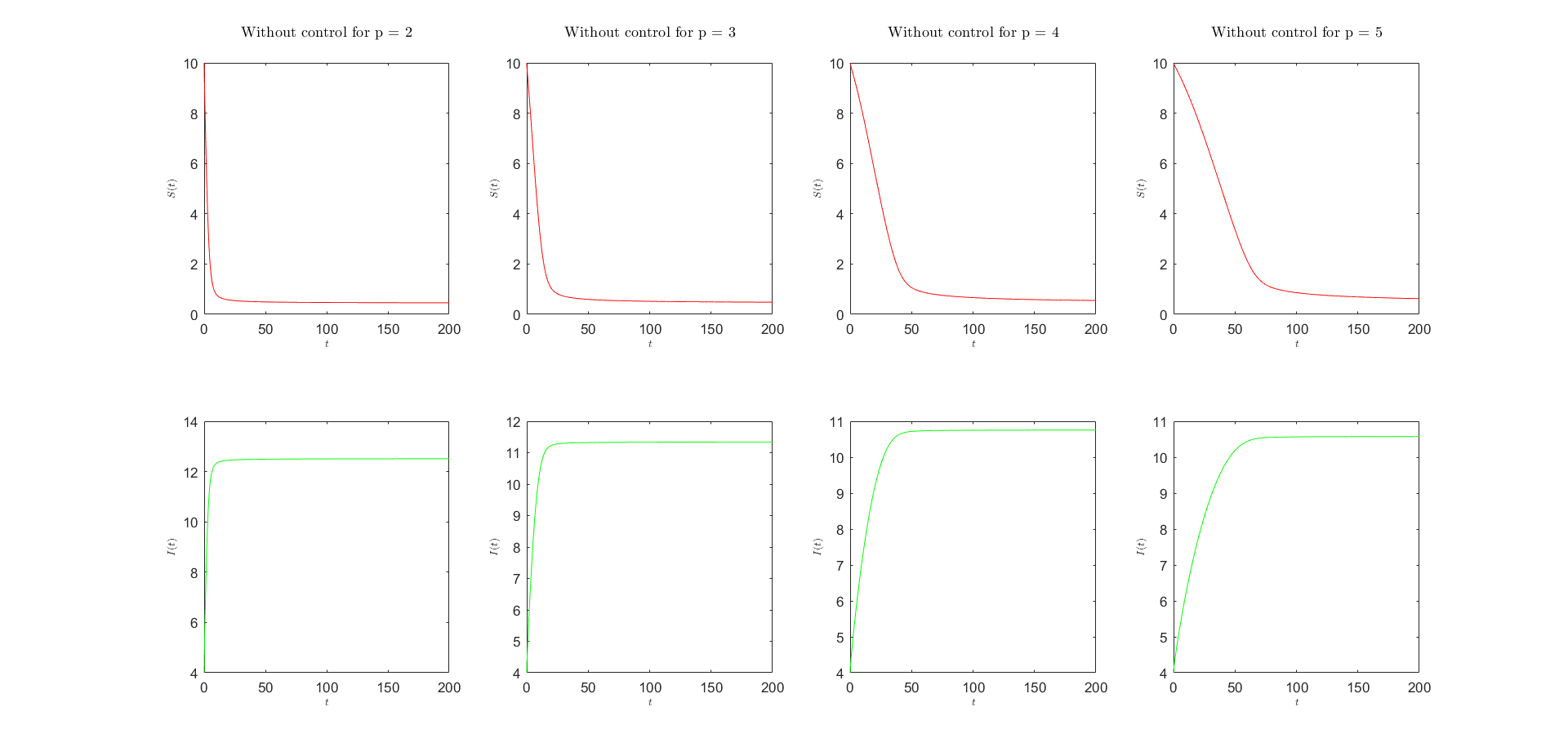}
\caption{Numerical simulations of the SI model \eqref{E3.1}--\eqref{E3.2} for $\alpha = 0.8$ without control.}\label{F4}
\end{figure}

\setlength{\tabcolsep}{0.5cm}
\begin{table}[htb]
\centering
\caption{The cost of $\mathscr{T}$ without control.}\label{Tab3}
\begin{tabular}{|c|c|c|c|}
\hline
\diagbox[width=19mm,height=4.5mm]{}{} & $\alpha = 0.8$ & $\alpha = 0.9$ & $\alpha = 1$\\
\hline
$p = 2$ & $1.2641e^{+06}$ & $1.2579e^{+06}$ & $1.2672e^{+06}$\\
\hline
$p = 3$ & $1.0124e^{+06}$ & $1.0021e^{+06}$ & $ 1.0174e^{+06}$\\
\hline
$p = 4$ & $8.6621e^{+05}$ & $8.3237e^{+05}$ & $8.8288e^{+05}$\\
\hline
$p = 5$ & $7.8946e^{+05}$ & $7.2172e^{+05}$ & $8.2216e^{+05}$\\
\hline
\end{tabular}
\end{table}

\clearpage
\subsection{Numerical approximations in the presence of vaccination}
The vaccination strategy is currently being implemented, with the assumption that vaccination will begin on the first day. 
From Figures~\ref{F5}, \ref{F6} and \ref{F7}, we can see that the vaccination strategy has spectacularly fought the spread of the epidemic,
where the proportion of infected individuals has been reduced and, conversely, the number of susceptible individuals is constantly increasing.
We can say that our proposed strategy is effective in controlling the spread of the epidemic.

\begin{figure}[H]
\centering
\includegraphics[height=8cm]{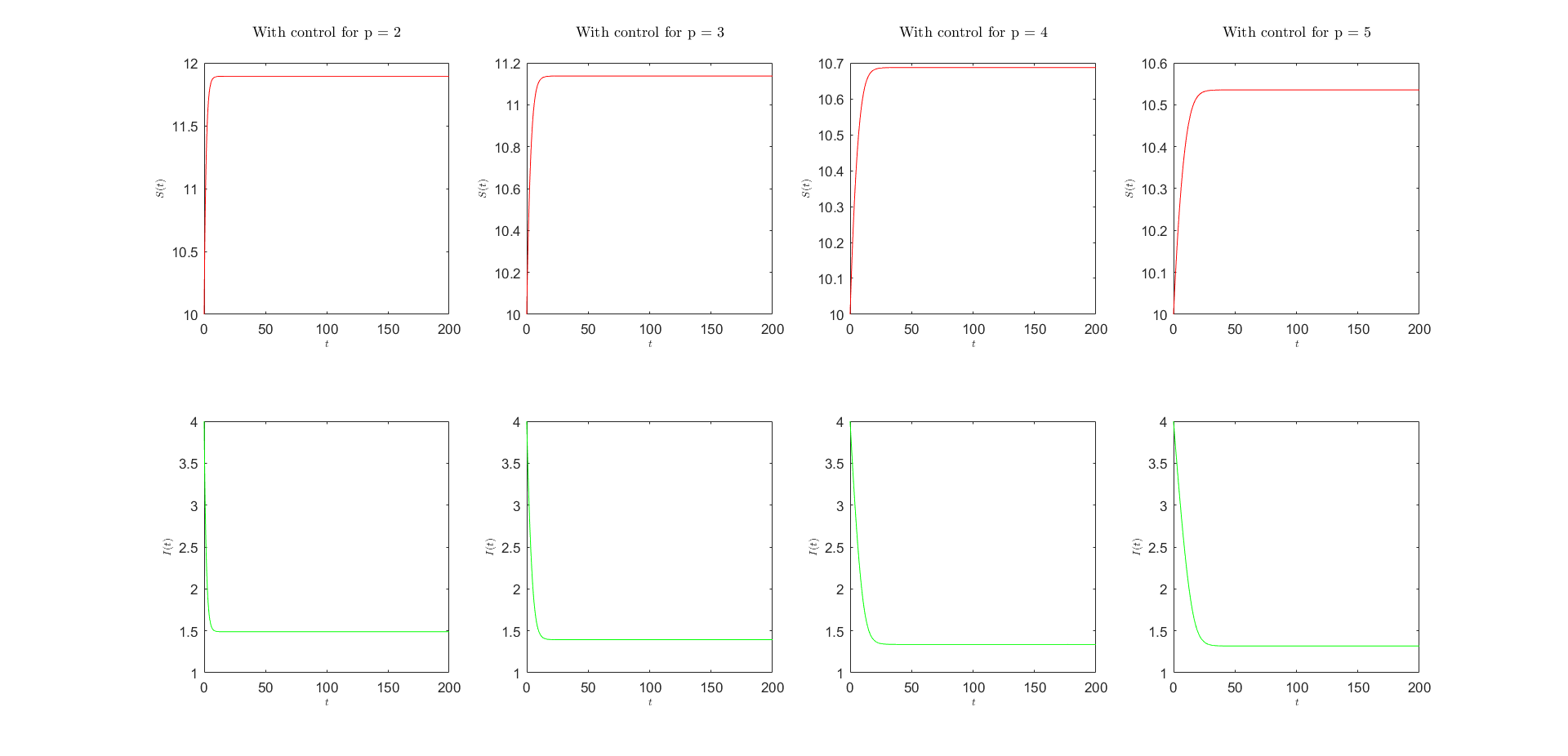}
\caption{Numerical simulations of the SI model \eqref{E3.1}--\eqref{E3.2} for $\alpha = 1$ with control.}\label{F5}
\end{figure}

\begin{figure}[H]
\centering
\includegraphics[height=8cm]{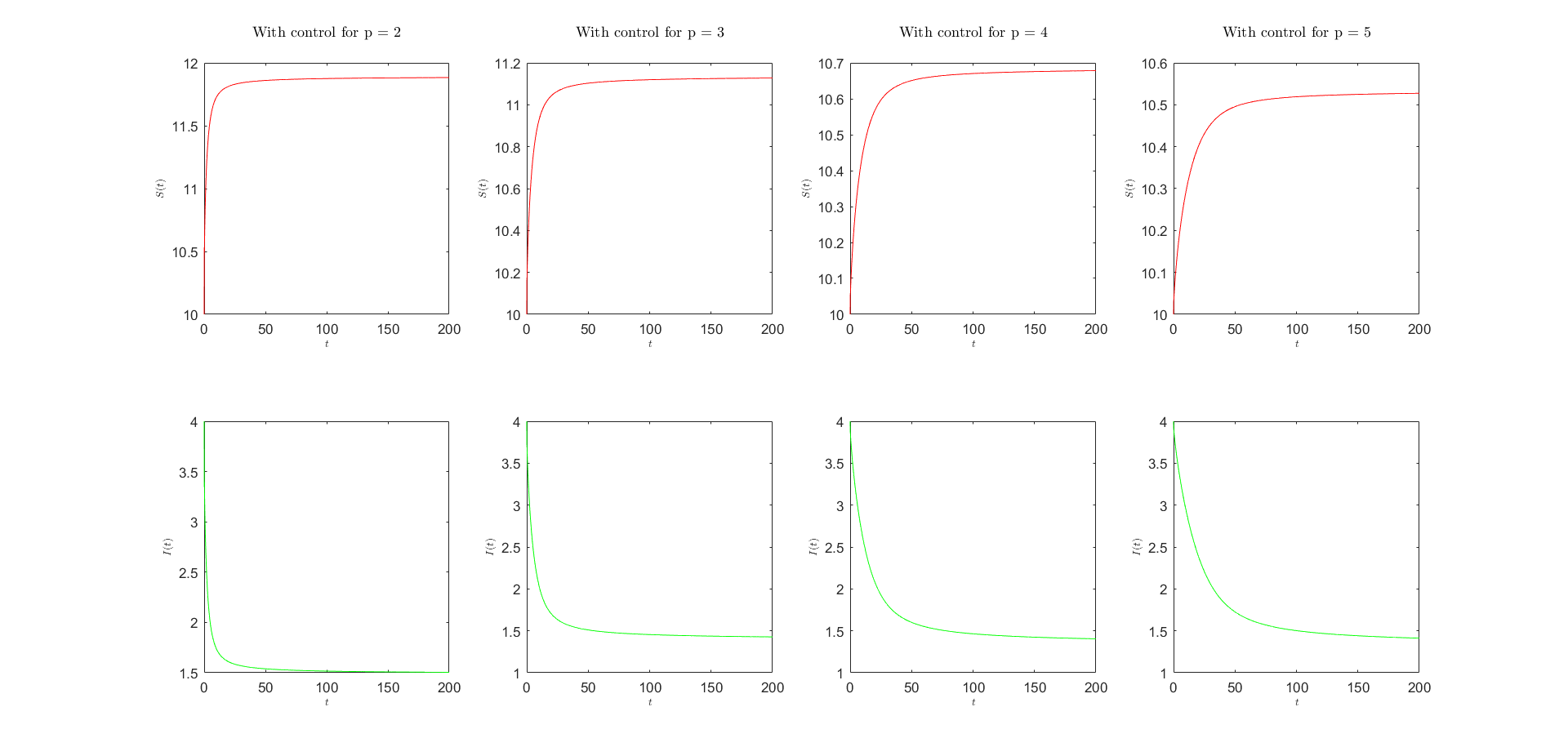}
\caption{Numerical simulations of the SI model \eqref{E3.1}--\eqref{E3.2} for $\alpha = 0.9$ with control.}\label{F6}
\end{figure}

\begin{figure}[H]
\centering
\includegraphics[height=8cm]{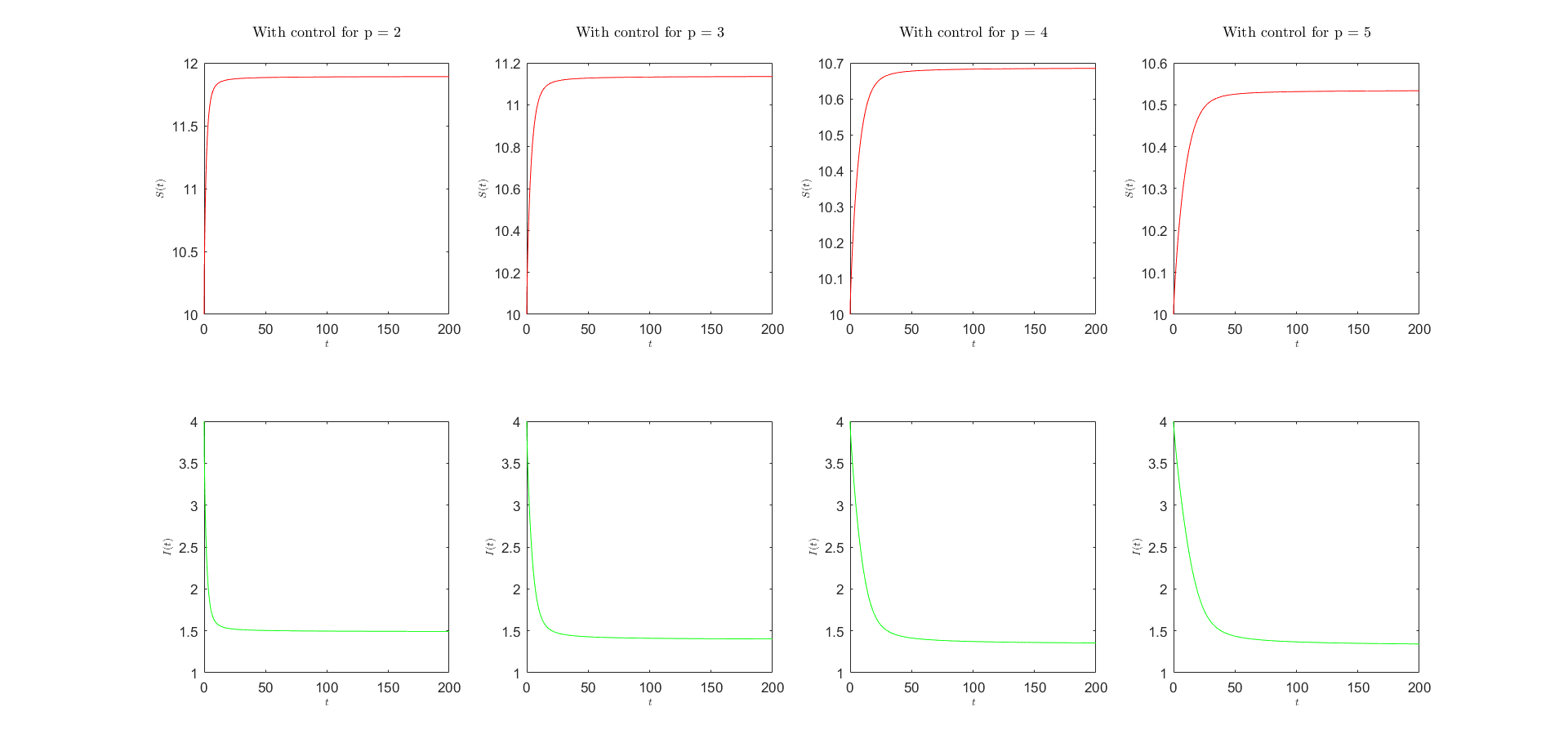}
\caption{Numerical simulations of the SI model \eqref{E3.1}--\eqref{E3.2} for $\alpha = 0.8$ with control.}\label{F7}
\end{figure}

\setlength{\tabcolsep}{0.5cm}
\begin{table}[H]
\centering
\caption{The cost of $\mathscr{T}$ with control.}\label{Tab4}
\begin{tabular}{|c||c|c|c|}
\hline
\diagbox[width=19mm,height=4.5mm]{}{} & $\alpha = 0.8$ & $\alpha = 0.9$ & $\alpha = 1$\\
\hline\hline
$p = 2$ & $2.9641e^{+04}$ & $2.3167e^{+04}$ & $4.1000e^{+04}$\\
\hline
$p = 3$ & $2.4244e^{+04}$ & $2.0200e^{+04}$ & $3.1549e^{+04}$\\
\hline
$p = 4$ & $1.9380e^{+04}$ & $1.8588e^{+04}$ & $2.2178e^{+04}$\\
\hline
$p = 5$ & $1.9117e^{+04}$ & $1.7859e^{+04}$ & $2.0106e^{+04}$\\
\hline
\end{tabular}
\end{table}

Note that, depending on the tables~\ref{Tab3} and \ref{Tab4}, we observe that $\mathscr{T}$ decreases under the effect of vaccination for different values of $p$ and $\alpha$. 
The value of $\mathscr{T}$ is optimal for $p = 5$ and $\alpha = 0.9$.
 
\section{Conclusion}\label{S8}
In this manuscript, we have presented an innovative application of epidemiological models,
where the interactions between susceptible and infected individuals are captured by a system of FDEs using the $\ABC$ fractional derivative and the $p$ Laplacian operator. 
Our investigation aims to contribute to more realistic models of disease spread in certain scenarios. 
We established the existence of a unique bounded non-negative solution to our biological model. 
We also proved the Hyers-Ulam stability and optimality conditions. 
We validated our theoretical findings through numerical simulations, where the results obtained show that the rapid spread of the disease is associated with $\alpha$ taking natural values and $p$ taking small values. 
Moreover, the slow spread of the disease results from the fractional and large values of both $\alpha$ and $p$. 
In conclusion, the spread of the disease has been successfully controlled.


\section*{Declarations}



\subsection*{CRediT author statement} 

\textit{A. Zinihi:} Conceptualization, 
Methodology, 
Software, 
Formal analysis, 
Investigation, 
Writing -- Original Draft, 
Writing -- Review \& Editing, 
Visualization. 

\textit{M. R. Sidi Ammi:} Conceptualization, 
Methodology, 
Validation, 
Formal analysis, 
Investigation,  
Writing -- Original Draft, 
Writing -- Review \& Editing,  
Supervision.

\textit{M. Ehrhardt:} Validation, 
Formal analysis, 
Investigation, 
Writing -- Original Draft, 
Writing -- Review \& Editing, 
Project administration.





\subsection*{Data availability} 

All information analyzed or generated, which would support the results of this work are available in this article.
No data was used for the research described in the article.

\subsection*{Conflict of interest} 
The authors declare that there are no problems or conflicts 
of interest between them that may affect the study in this paper.


\bibliographystyle{acm}
\bibliography{paper}

\begin{thebibliography}{10}

\bibitem{Abdeljawad2017}
{\sc Abdeljawad, T., and Baleanu, D.}
\newblock {Integration by parts and its applications of a new nonlocal
  fractional derivative with Mittag-Leffler nonsingular kernel}.
\newblock {\em J. Nonlin. Sci. Appl. 10}, 03 (2017), 1098--1107.

\bibitem{Agarwal2010}
{\sc Agarwal, R.~P., O’Regan, D., and Staněk, S.}
\newblock Positive solutions for {D}irichlet problems of singular nonlinear
  fractional differential equations.
\newblock {\em J. Math. Anal. Appl. 371}, 1 (2010), 57–68.

\bibitem{AltafKhan2020}
{\sc Altaf~Khan, M., Ismail, M., Ullah, S., and Farhan, M.}
\newblock {Fractional order SIR model with generalized incidence rate}.
\newblock {\em AIMS Mathematics 5}, 3 (2020), 1856–1880.

\bibitem{Anderson1991}
{\sc Anderson, R.~M., and May, R.~M.}
\newblock {\em {Infectious diseases of humans: Dynamics and control}}.
\newblock Oxford University Press, 1991.

\bibitem{Atangana2016}
{\sc Atangana, A., and Baleanu, D.}
\newblock {New fractional derivatives with nonlocal and non-singular kernel:
  Theory and application to heat transfer model}.
\newblock {\em Therm. Sci. 20}, 2 (2016), 763--769.

\bibitem{Bai2004}
{\sc Bai, C.~Z., and Fang, J.~X.}
\newblock The existence of a positive solution for a singular coupled system of
  nonlinear fractional differential equations.
\newblock {\em Appl. Math. Comput. 150}, 3 (2004), 611–621.

\bibitem{Bai2009}
{\sc Bai, Z., and Qiu, T.}
\newblock Existence of positive solution for singular fractional differential
  equation.
\newblock {\em Appl. Math. Comput. 215}, 7 (2009), 2761–2767.

\bibitem{Chai2012}
{\sc Chai, G.}
\newblock Positive solutions for boundary value problem of fractional
  differential equation with $p$-{L}aplacian operator.
\newblock {\em Boundary Value Problems 2012}, 1 (2012).

\bibitem{Chen2012}
{\sc Chen, T., Liu, W., and Hu, Z.}
\newblock A boundary value problem for fractional differential equation with
  $p$-{L}aplacian operator at resonance.
\newblock {\em Nonlin. Anal.: Real World Appl. 75}, 6 (2012), 3210–3217.

\bibitem{Din2021}
{\sc Din, A., Li, Y., Khan, F.~M., Khan, Z.~U., and Liu, P.}
\newblock {On analysis of fractional order mathematical model of Hepatitis B
  using Atangana–Baleanu Caputo (ABC) derivative}.
\newblock {\em Fractals 30}, 01 (2021).

\bibitem{Djida2018}
{\sc Djida, J.~D., Mophou, G., and Area, I.}
\newblock {Optimal control of diffusion equation with fractional time
  derivative with nonlocal and nonsingular Mittag-Leffler kernel}.
\newblock {\em J. Optim. Theory Appl. 182}, 2 (2018), 540–557.

\bibitem{Fernandez2018}
{\sc Fernandez, A., and Baleanu, D.}
\newblock {The mean value theorem and Taylor’s theorem for fractional
  derivatives with Mittag-Leffler kernel}.
\newblock {\em Adv. Diff. Eqs. 2018}, 1 (2018).

\bibitem{Forti1995}
{\sc Forti, G.~L.}
\newblock {Hyers-Ulam stability of functional equations in several variables}.
\newblock {\em Aequat. Math. 50}, 1–2 (1995), 143–190.

\bibitem{Gatto2021}
{\sc Gatto, N.~M., and Schellhorn, H.}
\newblock Optimal control of the {SIR} model in the presence of transmission
  and treatment uncertainty.
\newblock {\em Math. Biosci. 333\/} (2021), 108539.

\bibitem{Ghoshal2024}
{\sc Ghoshal, S.~S., Junca, S., and Parmar, A.}
\newblock Fractional regularity for conservation laws with discontinuous flux.
\newblock {\em Nonlin. Anal.: Real World Appl. 75\/} (2024), 103960.

\bibitem{Han2023}
{\sc Han, J., and Liu, C.}
\newblock Global weak solution for a chemotaxis {N}avier–{S}tokes system with
  $p$-{L}aplacian diffusion and singular sensitivity.
\newblock {\em Nonlin. Anal.: Real World Appl. 73\/} (2023), 103898.

\bibitem{Hilfer2000}
{\sc Hilfer, R.}
\newblock {\em Applications of fractional calculus in physics}.
\newblock World Scientific, 2000.

\bibitem{NSFD_Review}
{\sc Hoang, M.~T., and Ehrhardt, M.}
\newblock Differential equation models for infectious diseases: Mathematical
  modeling, qualitative analysis, numerical methods and applications.
\newblock IMACM preprint 24/02, University of Wuppertal, Germany, 2024.

\bibitem{Hosseininia2022}
{\sc Hosseininia, M., Heydari, M.~H., and Razzaghi, M.}
\newblock {Meshless local Petrov-Galerkin method for 2D fractional
  Fokker-Planck equation involved with the ABC fractional derivative}.
\newblock {\em Comput. Math. Appl. 125\/} (2022), 176–192.

\bibitem{Huo2015}
{\sc Huo, J., Zhao, H., and Zhu, L.}
\newblock {The effect of vaccines on backward bifurcation in a fractional order
  HIV model}.
\newblock {\em Nonlin. Anal.: Real World Appl. 26\/} (2015), 289–305.

\bibitem{Jena2021}
{\sc Jena, R.~M., Chakraverty, S., and Baleanu, D.}
\newblock {SIR epidemic model of childhood diseases through fractional
  operators with Mittag-Leffler and exponential kernels}.
\newblock {\em Math. Comput. Simul. 182\/} (2021), 514–534.

\bibitem{Keeling2007}
{\sc Keeling, M.~J., and Rohani, P.}
\newblock {\em Modeling infectious diseases in humans and animals}.
\newblock Princeton University Press, 2007.

\bibitem{Kermack1927}
{\sc Kermack, W.~O., and McKendrick, A.~G.}
\newblock A contribution to the mathematical theory of epidemics.
\newblock {\em Proc. Royal Soc. A 115}, 772 (1927), 700--721.

\bibitem{Khan2021}
{\sc Khan, A., Zarin, R., Humphries, U.~W., Akg\"{u}l, A., Saeed, A., and Gul,
  T.}
\newblock {Fractional optimal control of COVID-19 pandemic model with
  generalized Mittag-Leffler function}.
\newblock {\em Adv. Diff. Eqs. 2021}, 1 (2021).

\bibitem{Khan2018}
{\sc Khan, H., Chen, W., and Sun, H.}
\newblock {Analysis of positive solution and Hyers‐Ulam stability for a class
  of singular fractional differential equations with $p$‐Laplacian in Banach
  space}.
\newblock {\em Math. Meth. Appl. Sci. 41}, 9 (2018), 3430–3440.

\bibitem{Lin2010}
{\sc Lin, F., Muthuraman, K., and Lawley, M.}
\newblock An optimal control theory approach to non-pharmaceutical
  interventions.
\newblock {\em BMC Infect. Dis. 10}, 1 (2010).

\bibitem{Maamar2024}
{\sc Maamar, M., Ehrhardt, M., and Tabharit, L.}
\newblock {A Nonstandard Finite Difference Scheme for a Time-Fractional Model
  of Zika Virus Transmission}.
\newblock {\em Math. Biosci. Engrg. 21}, 01 (2024), 924--962.

\bibitem{mittag1903generalisation}
{\sc Mittag-Leffler, G.~M.}
\newblock {Une g{\'e}n{\'e}ralisation de l’int{\'e}grale de Laplace-Abel}.
\newblock {\em CR Acad. Sci. Paris (Ser. II) 137\/} (1903), 537--539.

\bibitem{Niimi2012}
{\sc Niimi, K., Harding, D. R.~K., Holmes, A.~R., Lamping, E., Niimi, M.,
  Tyndall, J. D.~A., Cannon, R.~D., and Monk, B.~C.}
\newblock {Specific interactions between the Candida albicans ABC transporter
  Cdr1p ectodomain and ad‐octapeptide derivative inhibitor}.
\newblock {\em Molecul. Microbiol. 85}, 4 (2012), 747–767.

\bibitem{Okazawa2002}
{\sc Okazawa, N., and Yokota, T.}
\newblock Global existence and smoothing effect for the complex
  {G}inzburg–{L}andau equation with p-{L}aplacian.
\newblock {\em J. Diff. Eqs. 182}, 2 (2002), 541–576.

\bibitem{Phukan2023}
{\sc Phukan, J., and Dutta, H.}
\newblock {Dynamic analysis of a fractional order SIR model with specific
  functional response and Holling type II treatment rate}.
\newblock {\em Chaos Solitons \& Fractals 175\/} (2023), 114005.

\bibitem{Podlubny1998}
{\sc Podlubny, I.}
\newblock {\em {Fractional differential equations: An introduction to
  fractional derivatives, fractional differential equations, to methods of
  their solution and some of their applications}}.
\newblock Elsevier, 1998.

\bibitem{Pu2018}
{\sc Pu, H., Li, S., Dong, M., Jiao, S., and Shang, Y.}
\newblock Numerical method for coupled thermal analysis of the regenerative
  cooling structure.
\newblock {\em J. Thermophys. Heat Transf. 32}, 2 (2018), 326–336.

\bibitem{Samko1993}
{\sc Samko, S.~G., Kilbas, A., and Marichev, O.}
\newblock {\em Fractional Integrals and Derivatives: Theory and Applications}.
\newblock Gordon and Breach, 1993.

\bibitem{SidiAmmi2023}
{\sc Sidi~Ammi, M.~R., Zinihi, A., Raezah, A.~A., and Sabbar, Y.}
\newblock Optimal control of a spatiotemporal {SIR} model with
  reaction{\textendash}diffusion involving $p$-{L}aplacian operator.
\newblock {\em Res. Phys. 52\/} (2023), 106895.

\bibitem{Ulam1964}
{\sc Ulam, S.~M.}
\newblock {\em A collection of mathematical problems}.
\newblock Interscience, New York, 1964.

\bibitem{Vong2013}
{\sc Vong, S.}
\newblock Positive solutions of singular fractional differential equations with
  integral boundary conditions.
\newblock {\em Math. Comput. Modell. 57}, 5–6 (2013), 1053–1059.

\bibitem{Wang2023}
{\sc Wang, L., Wang, J., Zhong, Q., and Cheng, K.}
\newblock {Concentration of solutions for a fractional relativistic
  Schr\"{o}dinger-Choquard equation with critical growth}.
\newblock {\em Nonlin. Anal. 230\/} (2023), 113233.

\bibitem{Zhao2023}
{\sc Zhao, K.}
\newblock {Generalized UH-stability of a nonlinear fractional coupling
  $({p}_{1}, {p}_{2})$-Laplacian system concerned with nonsingular
  Atangana-Baleanu fractional calculus}.
\newblock {\em J. Inequal. Appl. 2023}, 1 (2023).

\end{thebibliography}


\end{document}